%% file: arxiv.tex
\documentclass[review,hidelinks,onefignum,onetabnum]{siamart220329}

\input{ex_shared}

\ifpdf
\hypersetup{
  pdftitle={},
  pdfauthor={}
}
\fi

\newcommand{\RR}{\mathbb{R}}
\newcommand{\EE}{\mathbb{E}}
\newcommand{\PP}{\mathbb{P}}
\newcommand{\NN}{\mathbb{N}}

\newcommand{\eps}{\varepsilon}

\newcommand{\mtcc}{\mathcal{C}}

\newcommand{\mtce}{\mathcal{E}}

\newcommand{\mtcg}{\mathcal{G}}

\newcommand{\mtcs}{\mathcal{S}}

\newcommand{\mtcv}{\mathcal{V}}

\newcommand{\mtcx}{\mathcal{X}}

\newcommand{\Let}{: =}
\newcommand{\teL}{= :}

\newcommand{\bfl}{\mathbf{1}}
\newcommand{\bfo}{\mathbf{0}}
\newcommand{\bfx}{\mathbf{x}}
\newcommand{\tx}{\textup}
\newcommand{\tp}{\tx{T}}

\newcommand{\nr}{n_{\tx{r}}}
\newcommand{\ns}{n_{\tx{s}}}
\newcommand{\Psir}{\Psi^{(\tx{r})}}
\newcommand{\Psis}{\Psi^{(\tx{s})}}
\newcommand{\psir}{\psi^{(\tx{r})}}
\newcommand{\psis}{\psi^{(\tx{s})}}
\newcommand{\mtcvr}{\mtcv_{\tx{r}}}
\newcommand{\mtcvs}{\mtcv_{\tx{s}}}
\newcommand{\zs}{z^{(\tx{s})}}

\newcommand{\Xmtcg}{X^{\mtcg}}
\newcommand{\Qmtcg}{Q^{\mtcg}}
\newcommand{\Rmtcg}{R^{\mtcg}}
\newcommand{\barQg}{\bar{Q}^{\mtcg}}
\newcommand{\barRg}{\bar{R}^{\mtcg}}
\newcommand{\barMg}{\bar{M}^{\mtcg}}
\newcommand{\barUg}{\bar{U}^{\mtcg}}
\newcommand{\barLg}{\bar{L}^{\mtcg}}
\newcommand{\alphag}{\alpha^{\mtcg}}
\newcommand{\Wg}{W^{\mtcg}}
\newcommand{\barmtcg}{\bar{\mtcg}}
\newcommand{\barmtce}{\bar{\mtce}}
\newcommand{\xgn}{\bfx^{\mtcg,n}}
\newcommand{\EEg}{\EE_{\mtcg}}
\newcommand{\Smtcg}{S^{\mtcg}}
\newcommand{\Xst}{X^{*}}
\newcommand{\Qst}{Q^{*}}
\newcommand{\Rst}{R^{*}}
\newcommand{\barQst}{\bar{Q}^{*}}
\newcommand{\barRst}{\bar{R}^{*}}
\newcommand{\xstn}{\bfx^{*,n}}
\newcommand{\barMst}{\bar{M}^{*}}
\newcommand{\barUst}{\bar{U}^{*}}
\newcommand{\barLst}{\bar{L}^{*}}
\newcommand{\alphast}{\alpha^{*}}
\newcommand{\Wst}{W^{*}}

\newcommand{\trg}{\tx{RG}}
\newcommand{\trgs}{\tx{RG-S}}

\newcommand{\dir}{d_i^{(\tx{r})}}
\newcommand{\dis}{d_i^{(\tx{s})}}
\newcommand{\di}{d_i}
\newcommand{\Deltar}{\Delta_{\tx{r}}}
\newcommand{\Deltarr}{\Delta_{\tx{rr}}}
\newcommand{\Deltars}{\Delta_{\tx{rs}}}

\newcommand{\Deltasr}{\Delta_{\tx{sr}}}

\newcommand{\deltars}{\delta_{\tx{rs}}}

\allowdisplaybreaks[4]

\begin{document}

\maketitle

\begin{abstract}
We study concentration inequalities in gossip opinion dynamics over random graphs. In the model, a network is generated from a random graph model with independent edges, and agents interact pairwise randomly over the network. During the process, regular agents average neighbors' opinions and then update, whereas stubborn agents do not change opinions. To approximate the original process, we introduce a gossip model over an expected graph, obtained by averaging all possible networks generated from the random graph model. Using concentration inequalities, we derive high-probability bounds for the distance between the expected final opinion vectors over the random graph and over the expected graph. 
Leveraging matrix perturbation results, we show how such concentration can help study the effect of network structure on the expected final opinions in two cases: (i) When the influence of stubborn agents is large, the expected final opinions polarize and are close to stubborn agents' opinions. 
(ii) When the influence of stubborn agents is small, the expected final opinions are close to each other. With the help of concentration inequalities for Markov chains, we obtain high-probability bounds for the distance between time-averaged opinions and the expected final opinions over the expected graph. In simulation, we validate the theoretical findings, and study a gossip model over a stochastic block model that has community structure. 
\end{abstract}



\begin{figure*}[t]
	\centering
	\subfigure[\label{fig:exp_consensus}Perfect consensus.]{\includegraphics[width=0.23\textwidth]{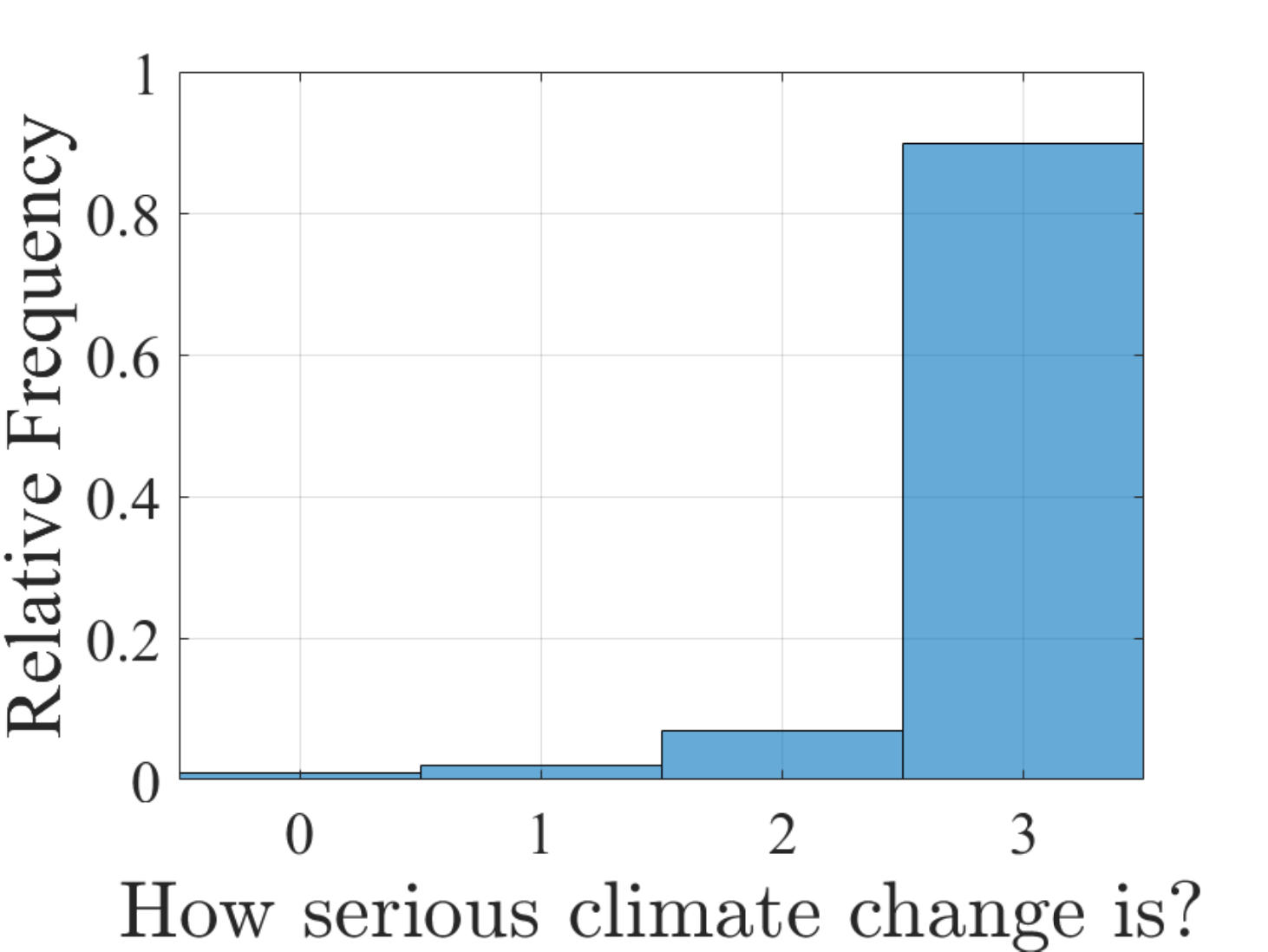}}	\subfigure[\label{fig:exp_polarization}Polarization.]{\includegraphics[width=0.23\textwidth]{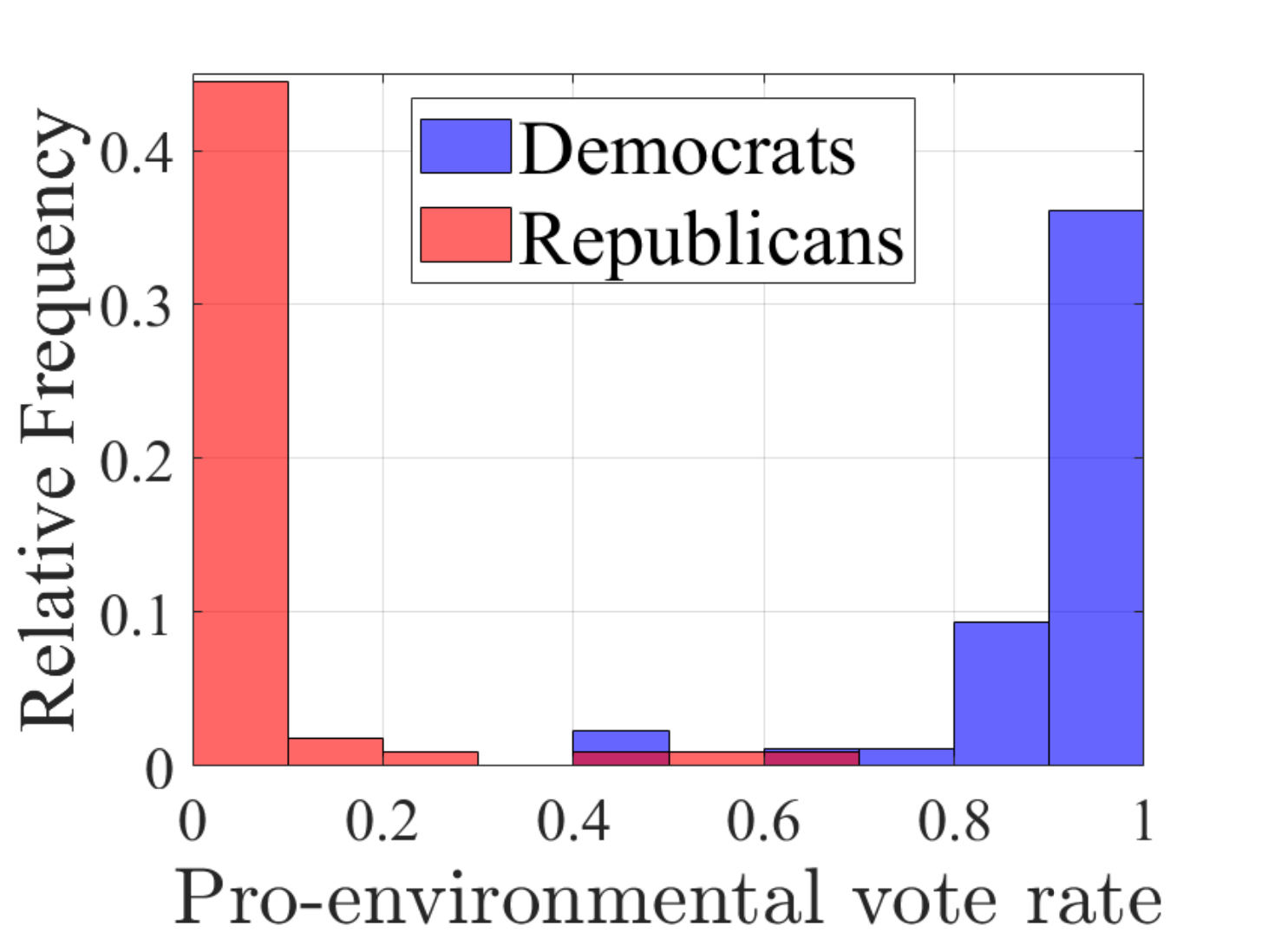} }    \subfigure[\label{fig:exp_trimodal}Clustering.]{\includegraphics[width=0.23\textwidth]{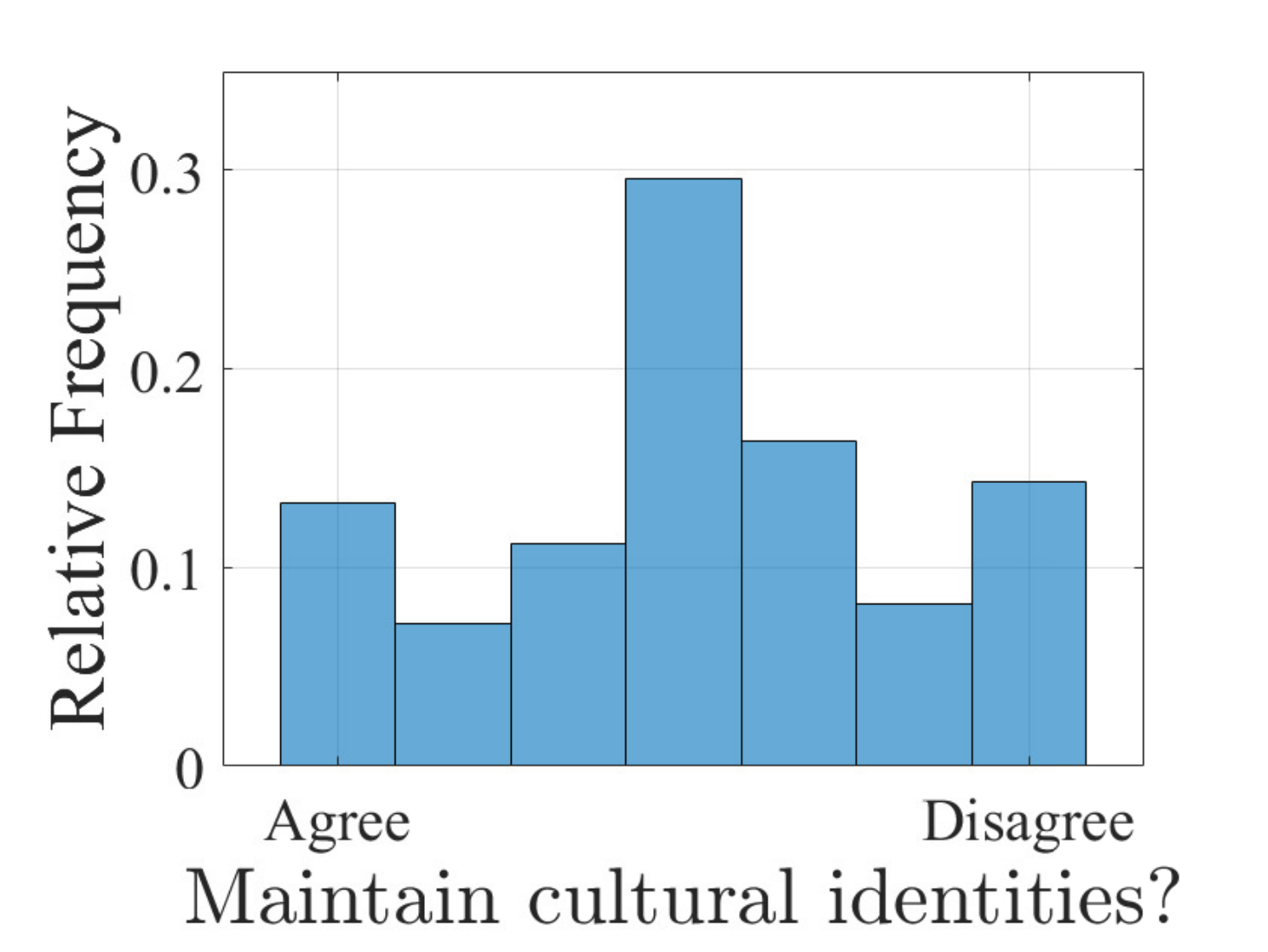} }	\subfigure[\label{fig:exp_pentamodal}Dissensus.]{\includegraphics[width=0.23\textwidth]{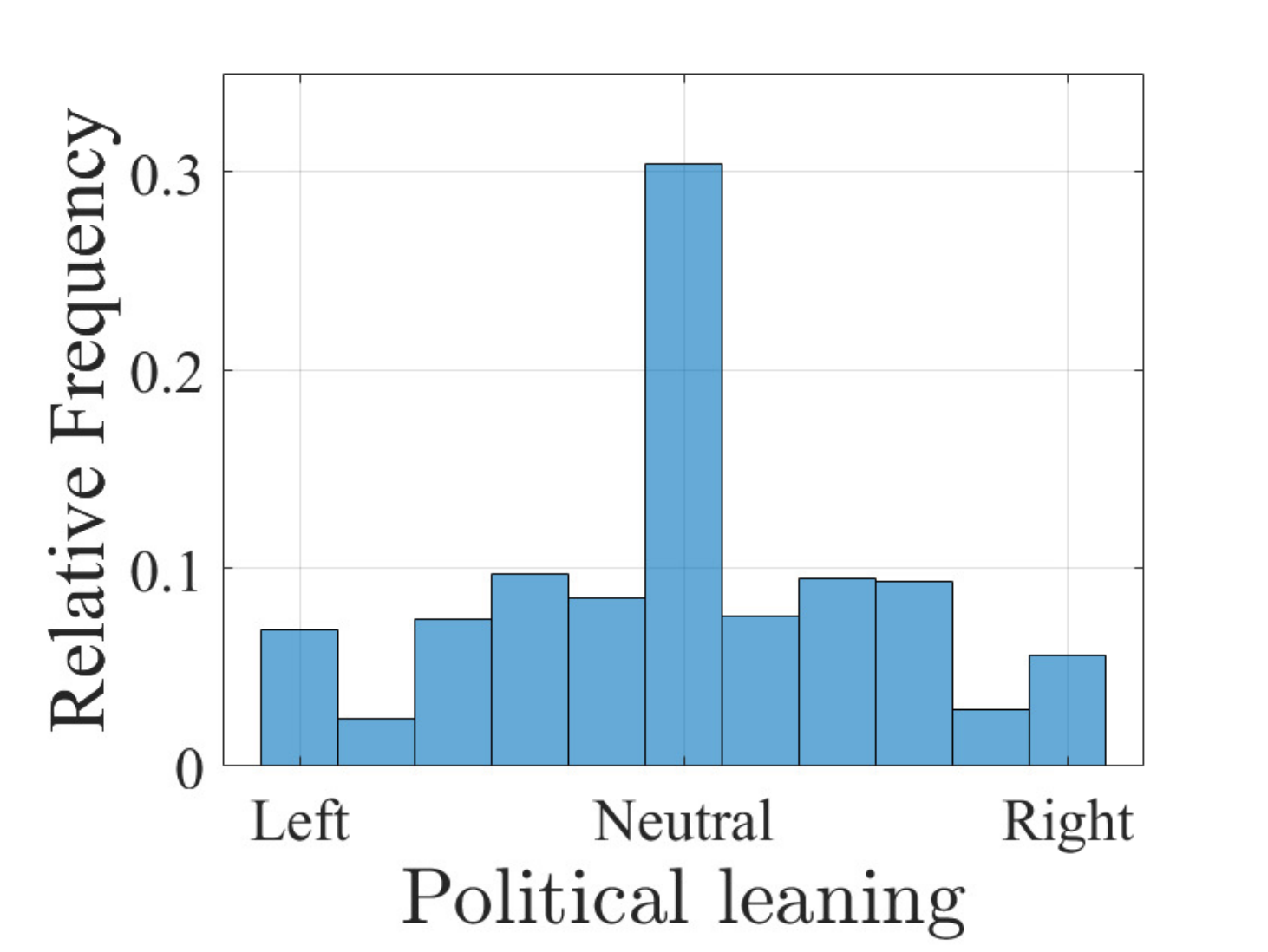} }
	\caption{\label{fig:exp}Different categories of opinion distributions (terminology from~\cite{devia2022framework}). (a) Perfect consensus in severity of climate change, where $0$ means ``don't know", $1$ ``not serious", $2$ ``fairly serious", and $3$ ``very serious". Almost all respondents in Spain of a survey regard climate change as a very serious problem~\cite{eurobarometer2020attitudes}. (b) Polarization of pro-enviormental votes on legislation from U.S. senators in 2015, where Democrats show high percentage of pro-environmental votes but Republicans show low percentage. The political elites hold extreme positions in line with their parties (Source: League of Conservation Voters)~\cite{dunlap2016political}.  (c) Clustering of opinions on whether people should maintain their distinct cultural identities~\cite{downey2001attitudinal}. Three clusters can be observed on the left, middle, and right, respectively. (d) Dissensus of French political opinions from European Social Survey 2012. Individual opinions are diverse, with most of them held by a non-negligible number of people.}
\end{figure*}

\section{Introduction}\label{sec:intro} 
Social opinion dynamics studies how interactions over networks shape individual opinion evolution, and has various applications~\cite{peralta2022opinion,zha2020opinion}. The last two decades have witnessed great developments in the study of opinion dynamics. Numerous mathematical approaches have been applied to modeling and analysis of such dynamics~\cite{castellano2009statistical,flache2017models,proskurnikov2017tutorial}. 
Most existing studies have focused on asymptotic behavior of opinion evolution and qualitative characterization of opinion distributions, such as consensus and polarization. An open problem is how to analyze the influence of specific network structure on the opinion evolution within a unified framework~\cite{flache2017models,proskurnikov2017tutorial}.
For example, community structure describes the property that subgroups of agents are connected densely with each other but loosely with other subgroups, which is often observed in reality~\cite{fortunato2010community,girvan2002community}. 
But how to quantify the relationship between the opinion evolution and the community structure is still not clear.
It is well-known that many network properties can be modeled by random graph models~\cite{bollobas1998random,newman2018networks,van2009random}. 
Combining random graph theory with the study of opinion dynamics can provide insight into linking microscopic agent updates to macroscopic system behaviors~\cite{flache2017models,proskurnikov2017tutorial} and offering quantitative predictions for real opinion evolution~\cite{friedkin2015problem}.

Let us consider a motivating example about diverse types of opinion distributions and how these distributions can be captured by a simple networked dynamical model.
\begin{example}\label{exmp:motivat}
Various types of opinion distributions can be observed in real-life scenarios. 
A common phenomenon is consensus, which occurs when individuals reach the same opinion on a particular issue, as shown in \Cref{fig:exp_consensus}. 
A group can diverge into two factions adopting opposite extreme views, which is known as polarization and illustrated by \Cref{fig:exp_polarization}.
Another type of opinion distributions is clustering, where individuals form two or more clusters, as demonstrated in \Cref{fig:exp_trimodal}.
Finally, dissensus can often be found in surveys~\cite{devia2022framework,flache2017models}, where most opinions are each held by a substantial number of people, as shown in \Cref{fig:exp_pentamodal}.

The rich opinion behaviors illustrated above can be captured by simple network models. 
In this paper, we study a gossip model with stubborn agents that is able to generate these behaviors. 
From this model, we can analytically quantify the influence of network structure and stubborn agents on final opinions of non-stubborn agents.
Consensus occurs if the stubborn agents have small influence, whereas polarization occurs if their influence is large. When the influence of stubborn agents is moderate, opinion distributions can exhibit multiple peaks, corresponding to community structure of the network. 
These results can be developed in a unified quantitative framework.
\end{example}

\subsection{Related Work}\label{subsec:relatedwork}

Individual opinions represent personal attitudes towards topics, events, or other persons, and can be modeled  by scalar or vector quantities~\cite{castellano2009statistical,proskurnikov2017tutorial}. Opinion dynamics describe how opinions evolve through interpersonal interactions. Continuous-state models are studied in this paper. The French--DeGroot (FD) model~\cite{degroot1974reaching
} shows how consensus is reached, where agents update by averaging their neighbors' opinions. Extensions of the model have been studied extensively~\cite{blondel2005convergence,cao2008reaching}. 
The gossip model generalizes the FD model by including random interactions between agents, and the model can exhibit various behavior such as consensus~\cite{boyd2006randomized,fagnani2008randomized}, disagreement, and opinion fluctuations~\cite{acemouglu2013opinion}. 
The Friedkin--Johnsen model~\cite{friedkin1990social} is another generalization of the FD model. It allows agents to be affected by their initial opinions, and generates long-term disagreement. Bounded confidence models (the Hegselmann--Krause model~\cite{hegselmann2002opinion} and the Deffuant--Weisbuch (DW) model~\cite{deffuant2000mixing}) explore how homophily influence shapes the opinion evolution. In these models, agents interact only with those who hold beliefs similar to them, and tend to form clusters. Models~\cite{altafini2012consensus,
shi2019dynamics}
with negative or antagonistic interactions, enlarging opinion difference, may end in polarization. 
In addition to interpersonal influences, stubborn agents also play crucial roles in opinion formation. These agents are assumed to never change opinions, representing opinion leaders and media sources. It has been shown that stubborn agents' opinions can determine the final opinions of the FD model~\cite{proskurnikov2017tutorial}. 
In the gossip model with stubborn agents, opinion fluctuations and long-term disagreement exist, but non-stubborn agents can have similar expected final opinions,  if the network is highly fluid~\cite{acemouglu2013opinion}. In contrast, for agents forming two communities connected to different stubborn agents, their final positions polarize if the influence of stubborn agents is large~\cite{como2016local}. The current paper revisits this classic model, and shows how to quantify the process in more detail with the help of random graph modeling.

Real networks often consist of numerous agents. To study large-scale group behavior, researchers have proposed macroscopic models which consider the evolution of opinion distributions. 
Eulerian approaches were introduced for analyzing bounded confidence models~\cite{canuto2012eulerian,kolarijani2021macroscopic,mirtabatabaei2014eulerian} and spatially distributed ordinary differential equations~\cite{nikitin2021continuation}. 
Graphon theory has been used recently for modeling heterogeneous large-scale networks, and the convergence of Euler approximations of mean-field games has been studied~\cite{bayraktar2022stationarity,caines2021graphon}.
Random graph theory is another framework for large-scale network modeling~\cite{bollobas1998random,bollobas2007phase,newman2018networks,newman2003structure,van2009random}. 
The field was founded by Erd{\H{o}}s and R{\'e}nyi~\cite{erdHos1960evolution} for studying probabilistic methods in graph theory. Since then various random graph models~\cite{barabasi1999emergence,watts1998collective} have been found to be useful in studying complex networks, such as small-world and scale-free networks~\cite{newman2018networks,newman2003structure,van2009random}. 
Random graphs have concentration properties; for instance, adjacency and Laplacian matrices can be close to their expectations~\cite{chung2011spectra,le2017concentration,tropp2015introduction}.
The influence of network structure on epidemics, dynamical systems, and search processes have been studied extensively~\cite{newman2018networks}. The stochastic block model (SBM) was introduced by~\cite{holland1983stochastic} to explain the generation of community structure.
Papers studying the influence of community structure on opinion evolution mainly focus on mean-field approximations and simulation (e.g., for the DW model~\cite{fennell2021generalized,gargiulo2010opinion}, the Sznajd model~\cite{si2009opinion}, three-state opinion models~\cite{oestereich2019three}, and a majority-vote model~\cite{peng2022majority}).

\subsection{Contribution}\label{subsec:contrib}
In this paper we study concentration in the gossip model over random graphs. We compare the model with a gossip model over an expected graph that is obtained by averaging all possible networks generated from the random graph model. We show that the expected final opinions of regular agents in the original model concentrate around those over the expected graph (\Cref{thm:concentration_states}). 
The distance between the two opinion vectors can be bounded by a quantity depending on the maximum and minimum expected degrees and stubborn-agent opinions. Using matrix perturbation theory, we study the effect of network structure and stubborn agents on the expected final opinions over the expected graph (\Cref{prop:profilexst}): (i)~When the influence of stubborn agents is large, regular agents hold final opinions close to stubborn agents.
(ii)~When the influence of stubborn agents is small, regular agents have final opinions close to each other. 
We obtain similar conclusions on the effect of network structure on the expected final opinions over the random graph (\Cref{thm_profileofxg}).
We also provide bounds for the distance between time-averaged opinions and the expected final opinions over the expected graph (\Cref{thm:concentration_states_time}). 

It is found that, unlike classic concentration results for adjacency and Laplacian matrices~\cite{chung2011spectra,le2017concentration,tropp2015introduction}, the concentration of expected final opinions depends on the relative magnitude of the maximum and minimum expected degrees in a random graph. 
Different from convergence and stability analysis~\cite{bauso2016opinion,bayraktar2022stationarity,caines2021graphon,canuto2012eulerian,mirtabatabaei2014eulerian}, the current paper quantifies the influence of network structure on opinion distributions. 
In particular, a unified framework is developed for approximating expected final opinions and time-averaged opinions (\Cref{thm:concentration_states,prop:profilexst,thm_profileofxg,thm:concentration_states_time}). 
Consequently, we can analyze the effect of network structure and stubborn agents on expected final opinions, provide conditions for the emergence of consensus~\cite{acemouglu2013opinion} and polarization~\cite{como2016local}, and establish correspondence between opinion evolution and community structure (see Section~\ref{sec:simul}). 
The gossip model over a two-community SBM is studied in the conference version~\cite{xing2022concentration}. The current paper studies concentration over general random graphs, explores the influence of network structure, and quantifies time-averaged opinions.

Because random graphs are widely used in modeling real networks~\cite{bollobas1998random,newman2018networks,newman2003structure,van2009random}, the current framework enables quantitative prediction of opinion evolution. 
More precisely, given a network, it is possible to establish random graph models from network properties, determine qualitative results for the evolution (e.g., whether polarization or consensus would happen), and then give high-probability bounds for the prediction. 
The obtained correspondence between community structure and agent opinions can inspire design of community detection methods based on state observations~\cite{schaub2020blind,xing2023community}. Suppose that the network is unknown but a trajectory of opinion evolution is available. It is possible to recover agent community labels by clustering agent states. Developing such a community detection algorithm is not done in this paper, but some further discussion on the problem is provided at the end of Section~\ref{sec:main}.

\subsection{Outline}\label{subsec:outline}

The paper is organized as follows. We describe the gossip model and random graph models in \Cref{sec:prelim}, and formulate the problem in \Cref{sec:problem}. \Cref{sec:mainresults} provides main results, \Cref{sec:simul} presents numerical experiments, and \Cref{sec:conclusions} concludes the paper. Proofs are provided in the Appendix.

\subsection*{Notation}

Denote the $n$-dimensional Euclidean space by $\mathbb{R}^n$, the set of $n\times m$ real matrices by $\mathbb{R}^{n\times m}$, the set of nonnegative integers by $\mathbb{N}$, and the set of positive integers by $\NN_+ = \NN\setminus\{0\}$. Denote the natural logarithm by $\log x$, $x>0$. 

Let $\mathbf{1}_n$ be the $n$-dimensional all-one vector, $e_i^{(n)}$ be the $n$-dimensional unit vector with $i$-th entry being one, $I_n$ be the $n\times n$ identity matrix, and $\bfo_{m,n}$ be the $m\times n$ all-zero matrix. For a vector $x\in \mathbb{R}^n$, denote its $i$-th entry by $x_i$, and for a matrix $A \in \mathbb{R}^{n\times n}$, denote its $(i,j)$-th entry by $a_{ij}$ or $[A]_{ij}$. 
Denote the Euclidean norm of a vector and the spectral norm of a matrix by $\|\cdot\|$.
Let $\rho(A)$ be the spectral radius of a square matrix  $A$.
For symmetric $A \in \RR^{n\times n}$, denote its eigenvalues by $\lambda_{\min}(A) \Let \lambda_1(A) \le \lambda_2(A) \le \cdots \le \lambda_n(A) \teL \lambda_{\max}(A)$. 
By $\diag(A_1,\dots,A_k)$ denote the diagonal or block diagonal matrix with $A_1$, $\dots$, $A_k$ on the diagonal.

The cardinality of a set $\mtcs$ is written as $|\mtcs|$. 
An event $A$ happens almost surely (a.s.) if $\PP\{A\}=1$. For a sequence of events $A_n$, we say $A_n$ happens with high probability (w.h.p.) if $\PP\{A_n\} \to 1$ as $n\to \infty$. 
For two sequences of real numbers, $f(n)$ and $g(n) > 0$, $n\in\NN$, we write $f(n) = O(g(n))$ if $|f(n)| \le C g(n)$ for all $n\in \NN$ and some $C > 0$, and write $f(n) = o(g(n))$ if $|f(n)|/g(n) \to 0$. 
Suppose $f(n) > 0$ for all $n\in\NN$. Write $f(n) = \omega(g(n))$ if $g(n) = o(f(n))$, and write $f(n) = \Omega(g(n))$ if $g(n) = O(f(n))$. 
For $x,y \in \RR$, denote their maximum by $x\vee y \Let\max\{x,y\}$ and their minimum by $x\wedge y\Let \min\{x,y\}$.
An undirected graph $\mtcg = (\mtcv, \mtce, A)$ has an agent set $\mtcv$, an edge set $\mtce$, and an adjacency matrix $A = [a_{ij}]$ with $a_{ij} = 1$ ($a_{ij} = 0$) if $\{i,j\} \in \mtce$ ($\{i,j\} \not\in \mtce$).

\section{Preliminaries}\label{sec:prelim}
In this section, we introduce network and dynamic models studied in the paper.  Section~\ref{subsec:graphmodel} describes a random graph model, and Section~\ref{subsec:opinionmodel} introduces the gossip model. We describe a random graph model with stubborn agents in Section~\ref{subsec:randomwithstubborn}, and the gossip model over random graphs in Section~\ref{subsec:gossipoverrandom}.

\subsection{Random Graph Model}\label{subsec:graphmodel}

In this subsection, we describe a random graph model motivated by capturing properties of real-world networks.
This random graph model assumes that edges in a network are generated independently~\cite{bollobas2007phase,chung2002connected}. 
\begin{definition}[Random graph model]\label{def:rg}
	Let $\mtcv = \{1,\dots,n\}$ with $n\in \NN_+$ be the set of agents and the symmetric matrix $\Psi = [\psi_{ij}] \in [0,1]^{n\times n}$ be the link probability matrix. In the random graph model $\trg(n, \Psi)$, an undirected random graph $\mtcg=(\mtcv,\mtce,A)$ without self-loops is constructed by adding an undirected edge $\{i,j\}$ to $\mtce$ with probability $\psi_{ij}$ independent of other agent pairs, for all $i,j\in \mtcv$ with $i\not=j$.
\end{definition}

The preceding definition is general and includes many classic examples.

\begin{example}~\label{exmp:rgs}\\\indent
	(i) When $\psi_{ij} \equiv \psi \in [0,1]$ for all $i,j \in \mtcv$, the random graph model is one version of the Erd{\H{o}}s--R{\'e}nyi model~\cite{van2009random}, where each edge exists with the same probability.
	
	(ii) Let $w = [w_1, \dots, w_n]^\tp \in \RR^n$ with $w_i \ge 0$ and $\max_i w_i^2 < \sum_k w_k$, and $\psi_{ij} = w_i w_j/ (\sum_{k} w_k)$. $\trg(n,\Psi)$ generates graphs with the expected degree sequence $w$~\cite{chung2002connected}.

        (iii) Assume that the agent set $\mtcv$ has $K \in \NN_+$ disjoint subsets called communities, $\mtcv_1$, $\dots$, $\mtcv_K$, and denote the community label of $i\in \mtcv_k$ by $\mtcc_i = k$, $1\le k \le K$. 
        Let the symmetric matrix $\Pi = [\pi_{ij}] \in [0,1]^{K\times K}$ be the link probability matrix for edges within and between communities. 
        $\trg(n,\Psi)$ with $\psi_{ij} = \pi_{\mtcc_i\mtcc_j}$, $i\not= j$, and $\psi_{ii} = 0$ is the SBM~\cite{holland1983stochastic}  that intuitively shows the formation of community structure.
\end{example}

\subsection{Gossip Model with Stubborn Agents}\label{subsec:opinionmodel}

In this subsection, we introduce the gossip model with stubborn agents and discuss its basic properties.  

A gossip model with stubborn agents (we call it ``the gossip model'' hereafter for short) is a random process evolving over a graph $\mtcg = (\mtcv, \mtce, A)$. 
The agent set $\mtcv$  contains regular agents $\mtcvr = \{1, \dots, \nr\}$ and stubborn agents $\mtcvs = \{1+\nr, \dots, \ns + \nr\}$, and the network size is $n = |\mtcv| = \nr+\ns$.
A regular agent~$i$ has opinion $X_i(t) \in \RR$ at time $t \in \NN$. 
A stubborn agent~$j$ has opinion $\zs_j$, and never changes it. 
Stacking the opinions, we denote the opinion vector of regular agents at time $t$ by $X(t) \in \RR^{\nr}$ and that of stubborn agents by $\zs \in \RR^{\ns}$ (for simplicity, we use $\zs_j$ to represent the opinion of~$j$, instead of $\zs_{j-\nr}$).
At each time, an edge is selected, and the two corresponding agents interact. 
The selection is modeled by an interaction probability matrix $W = [w_{ij}] \in \RR^{n\times n}$ depending on the adjacency matrix $A$, where $w_{ij} = w_{ji} = a_{ij}/\alpha$ and $\alpha = \sum_{i=1}^n \sum_{j=i+1}^n a_{ij}$ is the number of edges. An edge $\{i,j\}$ is selected with probability $w_{ij}$, independently of previous update. 
The two chosen agents are the only agents to update at time $t$.
If both~$i$ and~$j$ are regular, then $X_i(t+1) = X_j(t+1) = (X_i(t) + X_j(t))/2.$
If one of them is stubborn, say $j$, then $i$ updates as $X_i(t+1) = (X_i(t) + \zs_j)/2$. 
The update rule can be written as
\begin{align}\label{eq:gossipmodelo}
	X(t+1) = Q(t) X(t) + R(t) \zs.
\end{align}
Here $\{[Q(t)~R(t)]\}$ is a sequence of independent and identically distributed random matrices such that with probability~$w_{ij}$
\begin{small}
\begin{align}\label{eq:def_QR}
	[Q(t)~R(t)] = \begin{cases}
    [I_{n_r} - \frac12 (e_i^{(\nr)} - e_j^{(\nr)} )(e_i^{(\nr)}  - e_j^{(\nr)}) ^\tp,~\bfo_{\nr,\ns}], 
     &\tx{if } i, j \in \mtcvr,\\
    [I_{n_r} - \frac12 e_i^{(\nr)}(e_i^{(\nr)})^\tp,~\frac12 e_i^{(\nr)} (e_j^{(\ns)})^\tp], 
     &\tx{if } i \in \mtcvr,  j \in \mtcvs,
    \end{cases}
\end{align} 
\end{small}
where we use $e_j^{(\ns)}$ to represent $e^{(\ns)}_{j-\nr}$ for $j \in \mtcvs$ for notation simplicity.

Denote the expected interaction matrices by $\bar{Q} \Let \EE\{Q(t)\}$ and $\bar{R} \Let \EE\{R(t)\}$.
The following results~\cite{acemouglu2013opinion} (the paper studies the model in continuous time; see e.g.,~\cite{ravazzi2015ergodic,xing2023community} for analysis of discrete-time versions) indicate that the expected final opinions depend on the expected interaction matrices and opinions of stubborn agents.
\begin{proposition}[Stability and limit theorems]\label{prop:stability}
Suppose that $\mathcal{G}$ is connected and has at least one stubborn agent. The following results hold for the gossip model~\cref{eq:gossipmodelo}.\\\indent
(i) The model has a unique stationary distribution $\pi$ with mean $\bfx$, and $X(t)$ converges in distribution to $\pi$ as $t \to \infty$. The expected final opinions $\bfx$ satisfy that
\begin{align}\label{eq_expectation_limit}
    \bfx = \lim_{t \to \infty} \EE\{X(t)\} = (I-\bar{Q})^{-1}\bar{R} \zs.
\end{align}\indent
(ii) Denote the time-averaged opinions by $S(t) := \frac1t \sum_{i = 0}^{t - 1} X(i)$. Then \[ \lim_{t \to \infty} S(t) = \bfx \tx{ a.s.} \]
\end{proposition}
The results show that agent opinions converge in distribution to a stationary distribution, although they may fluctuate a.s.~\cite{acemouglu2013opinion}. Also, the time-averaged opinion vector $S(t)$ converges to $\mathbf{x}$, which characterizes the average final positions of regular agents.

\subsection{Random Graphs with Stubborn Agents}\label{subsec:randomwithstubborn}

To study the interplay between network structure and stubborn agents, we introduce the following definition of random graph model with stubborn agents. 
\begin{definition}[Random graph with stubborn agents, RG-S]\label{def:rg_s}~\\\indent
	Let $\mtcvr = \{1, \dots, \nr\}$ be the set of regular agents, $\mtcvs = \{1+\nr,\dots, \nr+\ns\}$ be the set of stubborn agents, and $n = \nr + \ns$ be the network size, where $\nr,\ns \in \NN_+$. Let the symmetric matrix $\Psir = [\psir_{ij}] \in [0,1]^{\nr\times \nr}$ be the link probability matrix for edges between regular agents, and $\Psis = [\psis_{ij}] \in [0,1]^{\nr\times \ns}$ be the link probability matrix for edges between regular and stubborn agents. 
 
    In the random graph model with stubborn agents $\trgs(\nr,\ns,\Psir,\Psis)$, a random graph $\mtcg=(\mtcv,\mtce,A)$ with $\mtcv = \{1,\dots,n\}$ is constructed according to the following rule: (i) A random graph for the regular agents is generated from $\trg(\nr,\Psir)$. \\(ii) For each regular agent $i \in \mtcvr$ and stubborn agent $j \in \mtcvs$, the edge $\{i,j\}$ is added to $\mtce$ with probability $\psis_{i,j-\nr}$, independent of other agent pairs.
\end{definition}
The RG-S includes stubborn agents in the network, and the link probability matrix $\Psis$ captures the influence strength of stubborn agents on regular agents. 

\subsection{Gossip Model over Random Graphs}\label{subsec:gossipoverrandom}
The previous subsections described the random graph models and the gossip model. In this subsection, we bring these models together. Suppose that a random graph $\mtcg$ is constructed from an RG-S, and over a realization of $\mtcg$ the gossip model takes place: 
\begin{align}\label{eq:gossipmodel}
	\Xmtcg(t+1) = \Qmtcg(t) \Xmtcg(t) + \Rmtcg(t) \zs, 
\end{align}
where $\Xmtcg(t)$ is the opinion vector and the superscript $\mtcg$ highlights the dependence of the process on $\mtcg$. Here $[\Qmtcg(t)~\Rmtcg(t)]$ has the expression given in~\eqref{eq:def_QR} but its distribution is defined by the interaction probability matrix $\Wg = A/\alphag$, where $A$ is the adjacency matrix of $\mtcg$ and $\alphag$ is the number of edges in $\mtcg$.

\begin{figure}[tbp]
  \centering
  \includegraphics[scale=0.32]{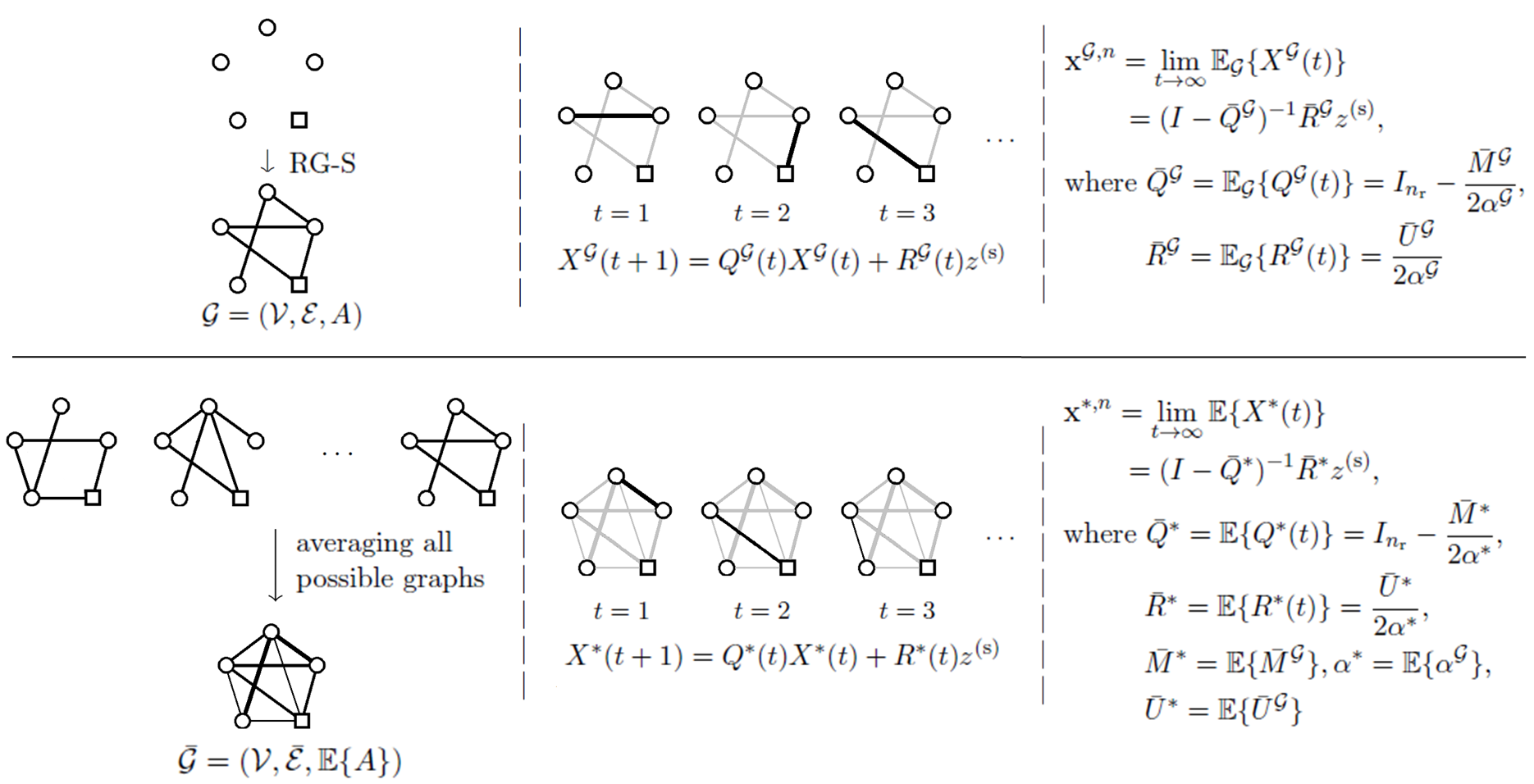}
  \caption{Illustration of a gossip model over an RG-S and a gossip model over an expected graph. On the top left of the figure, a random graph $\mtcg$ is constructed from an RG-S. Circles and squares represent regular and stubborn agents, respectively. 
  On the top middle, a gossip model evolves over $\mtcg$, where a single existing edge is selected at each time. On the top right, the expression of the expected final opinion vector is given. On the bottom left, the expected graph $\barmtcg$ is obtained by averaging the random graph $\mtcg$. On the bottom middle, a gossip model evolves over the expected graph, where an edge is selected with probability proportional to its weight in the expected adjacency matrix. On the bottom right, the expression of the expected final opinion vector over the expected graph is given.}
  \label{fig:illustration}
\end{figure}

Denote the expected interaction matrices by $\barQg \Let \EEg\{ \Qmtcg (t) \}$ and $\barRg \Let \EEg $ $\{\Rmtcg(t)\}$ (they are conditional expectations). If $(I - \barQg)^{-1}$ exists, the expected final opinion vector of the model can be written as
\begin{align}\label{eq:bfxg}
	\xgn \Let \lim_{t\to \infty} \EEg\{\Xmtcg(t)\} = (I - \barQg)^{-1} \barRg \zs, 
\end{align}
where  we use the superscripts $\mtcg$ and $n$ to indicate that the expected final opinions depend on the random graph $\mtcg$ and the network size $n$. 

To study behavior of the gossip model, we introduce a reference without network randomness. By averaging the random graph $\mtcg = (\mtcv, \mtce, A)$, we obtain the expected graph $\barmtcg = (\mtcv, \barmtce, \EE\{A\})$, where $\EE\{A\}$ is the expected adjacency matrix. Define a gossip model over this weighted graph $\barmtcg$ as follows.
\begin{definition}[Gossip model over expected graph]\label{def:gossip_expected}~\\\indent
    Consider a random graph model $\trgs(\nr,\ns,\Psir,\Psis)$ and its expected graph $\barmtcg = (\mtcv, \barmtce, \EE\{A\})$ obtained by averaging all graphs generated from the $\trgs$.
	The gossip model over the expected graph  is the following model that evolves over $\barmtcg$.
    \begin{align*}
	\Xst(t+1) = \Qst(t) \Xst(t) + \Rst(t) \zs, 
    \end{align*}
    where $\Xst(t)$ is the opinion vector, and $[\Qst(t)~\Rst(t)]$ has the same expression as in~\eqref{eq:def_QR} but its distribution is defined by the interaction probability matrix $\Wst = \EE\{A\}/\alphast$. Here $\alphast = \sum_{i=1}^n \sum_{j=i+1}^n \EE\{a_{ij}\}$ is the weight sum of the expected graph.
\end{definition}
Denote the expected interaction matrices by $ \barQst \Let \EE\{\Qst(t)\} $ and $\barRst \Let \EE\{\Rst(t)\}$. The expected final opinions of the model over the expected graph can be written as
\begin{align}\label{eq:bfx*}
	\xstn \Let \lim_{t\to\infty} \EE\{\Xst(t)\} =  (I - \barQst)^{-1} \barRst \zs.
\end{align}
In what follows we briefly explain the relations between quantities in the gossip model over the RG-S and those over the expected graph.
From~\eqref{eq:bfxg} we know that $\barQg$ and $\barRg$ determine the expected final opinion vector $\xgn$.
The expected interaction matrices $\barQg$ and $\barRg$ over the RG-S can be shown to have the following expressions: $\barQg = I_{\nr} - \barMg/(2\alphag)$ and $\barRg = \barUg/(2\alphag)$, where $\alphag$ is the number of edges in $\mtcg$,
\begin{small}
 \begin{align}\label{eq:barMg}
     \barMg &\Let 
	\begin{bmatrix}
		d_1 & - a_{12} & \dots & -a_{1,\nr} \\
		-a_{21} & d_2 & & \vdots \\
		\vdots & & \ddots & -a_{\nr-1,\nr}\\
		-a_{\nr,1} & \dots & -a_{\nr,\nr-1} & d_{\nr}
	\end{bmatrix},~
	\barUg \Let
	\begin{bmatrix}
		a_{1,\nr+1} & \dots & a_{1n} \\
		\vdots & & \vdots \\
		a_{\nr,\nr+1} & \dots & a_{\nr,n}
	\end{bmatrix},
\end{align}
\end{small}
\hspace{-1.5mm}and $d_i$ is the degree of the agent~$i$. Note that $\barMg$ and $\barUg$ depend on the adjacency matrix $A$ of $\mtcg$. For the gossip model over the expected graph, its expected final opinion $\xstn$ is determined by the expected interaction matrices $\barQst$ and $\barRst$. The two matrices are related to the expectations of $\barMg$, $\barUg$, and $\alpha$, i.e., $\barQst = I_{\nr} - \barMst/(2\alphast)$ and $\barRst = \barUst/(2\alphast) = \Psis/(2\alphast)$, where 
\begin{align}\label{eq:barMst}
    \barMst \Let \EE\{\barMg\},~ \barUst \Let \EE\{\barUg\} = \Psis, ~ \alphast = \EE\{\alphag\}.
\end{align}
\Cref{fig:illustration} summarizes the relations between the aforementioned quantities, and illustrates the gossip models over the $\trgs$ and the expected graph.

\section{Problem Formulation}\label{sec:problem}

This section formulates the problems of interest.

The first problem that we consider is when the expected final opinion vector $\xgn$ concentrates around the expected final opinion vector over the expected graph $\xstn$:

\textbf{Problem~$1$.} Given an RG-S and the gossip model~\eqref{eq:gossipmodel}, provide high probability bounds for the distance $\|\xgn - \xstn\|$.

Random graph models have concentration properties~\cite{chung2011spectra,tropp2015introduction,vershynin2018high}. For example, the eigenvalues of the adjacency
matrix of a random graph with independent edges concentrate around those of the expected graph, and the concentration error depends on the maximum expected degree~\cite{chung2011spectra}. Concentration inequalities can also be used in deriving degree conditions for connectivity of random graphs~\cite{tropp2015introduction}. 
Problem~$1$ arises naturally from these observations, but concerns the concentration of expected final opinions, rather than the random graph. The problem is addressed by \Cref{thm:concentration_states} in Section~\ref{sec:profile}, where conditions for $\xgn$ concentrating around $\xstn$ are given.

The second problem is to provide conditions for polarization or consensus of $\xgn$:

\textbf{Problem~$2$.} Given an RG-S and the gossip model~\eqref{eq:gossipmodel}, provide conditions for \\\indent
(i) the entries of $\xgn$ are close to opinions of stubborn agents,\\\indent
(ii) the entries of $\xgn$ are close to each other.

This problem concerns how network structure and stubborn agents shape the profile of the expected final opinions $\xgn$. Note that $\EE\{A\}$ has a simpler form than $A$, so it is easier to characterize $\xstn$ (\Cref{prop:profilexst}). Then using the solution to Problem~$1$, we are able to address Problem~$2$ in \Cref{thm_profileofxg}. When the network has community structure, according to \Cref{thm:concentration_states}, the expected final opinions can have clusters in line with the communities, which is illustrated in Section~\ref{sec:simul}. In this way we address the problem presented in \Cref{exmp:motivat} for the gossip model.

Finally, we derive bounds for the distance between the time-averaged opinions $\Smtcg(t)$ $= (\sum_{i=0}^{t-1} \Xmtcg(i))/t$ and the expected final opinions over the expected graph~$\xstn$:

\textbf{Problem~$3$.} Given an RG-S and the gossip model~\eqref{eq:gossipmodel}, provide high probability bounds for the distance $\|\Smtcg(t) - \xstn\|$.

This problem is important because only agent opinions can be observed in practice, rather than the expected opinions. From \Cref{prop:stability} we know that it is possible to use time-averaged opinions to estimate the expected opinions. 
Studying this problem can help us understand how network structure and stubborn agents affect transient behavior of the process. The result is given by \Cref{thm:concentration_states_time} in Section~\ref{sec:transient}.

\section{Main Results}\label{sec:mainresults}
\label{sec:main}

In this section, we first study the expected final opinions of the gossip model, by comparing them with those over the expected graph. We then investigate the behavior of time-averaged opinions.

\subsection{Concentration of Expected Final Opinions}\label{sec:profile}

In this subsection, we study properties of the expected final opinions $\xgn$. \Cref{thm:concentration_states} shows that the distance $\|\xgn - \xstn\|$ can be bounded by a term depending on maximum and minimum expected degrees of the RG-S with high probability. 
Next, we study in \Cref{prop:profilexst} how $\xstn$ is influenced by network structure and stubborn agents. Finally, we characterize the profile of $\xgn$ in \Cref{thm_profileofxg} by combining \Cref{thm:concentration_states,prop:profilexst}.

To begin with, we introduce the following notations. For an agent $i \in \mtcv$, we refer to the number of regular agents connected to~$i$ as its regular degree (denoted as $\dir$), and refer to  the number of stubborn agents connected to~$i$ as its stubborn degree (denoted as $\dis$). The degree of~$i$ is the sum of its regular and stubborn degrees, i.e., $\di = \dir + \dis$. 
The following quantities of the expected graph will be used frequently in the analysis. Let
\renewcommand\labelitemi{\tiny$\bullet$}
\begin{itemize}[leftmargin=*]
    \item  $\Deltar \Let \max\limits_{i \in \mtcvr}\{\EE\{\di\}\}$ be the maximum expected degree of regular agents,
    \item $\Deltarr \Let \max\limits_{i \in \mtcvr}\{\EE\{\dir\}\}$~be~the~maximum~expected~regular~degree~of~regular~agents,
    \item $\Deltars \Let \max\limits_{i \in \mtcvr}\{\EE\{\dis\}\}$ be the maximum expected stubborn degree of regular agents,
    \item $\Deltasr \Let \max\limits_{i \in \mtcvs}\{\EE\{\dir\}\}$ be the maximum expected regular degree of stubborn agents,
    \item $\deltars \Let \min\limits_{i \in \mtcvr}\{\EE\{\dis\}\}$ be the minimum expected stubborn degree of regular agents.
\end{itemize}

Assumptions of the main results are given below. The first assumption ensures large enough minimum expected stubborn degree $\deltars$, whereas the second assumption states lower bounds for the smallest eigenvalue of $\barMst$, given in~\eqref{eq:barMst}, and for maximum expected degrees $\Deltar$, $\Deltars$, and $\Deltasr$. The third assumption ensures that the gossip models over the RG-S and over the expected graph start with the same initial condition. The last assumption gives a lower bound for the number of regular agents.

\begin{assumption}\label{asmp:main}Assume that the following conditions hold.~\\\indent
	(i.1) $\deltars > 8\log n$.\\\indent
	(i.2) $\lambda_1(\barMst) > 4\sqrt{\Deltar \log n}$, $\Deltar \ge \log n$, and $\Deltars \vee \Deltasr \ge \log n$.\\\indent
	(ii) Both the gossip model over the RG-S and the gossip model over the expected graph have the same initial condition $X(0)$ and stubborn-agent opinions $\zs$. In addition, $\max_{i\in\mtcvr} \{|X_i(0)|\} \vee \max_{j\in\mtcvs}\{|\zs_j|\} \le c_x$ for some constant $c_x > 0$.\\\indent
        (iii) There exists a constant $c_{\tx{r}} \in (0,1)$ such that the proportion of regular agents $r_0 \Let \nr/n > c_{\tx{r}}$ for all $n \in \NN^+$.
\end{assumption}

\begin{remark}
    The condition~(i.1) requires that every regular agent has positive probability connected to some stubborn agent, whereas the condition~(i.2) allows the existence of regular agents not connected to any stubborn agents. Note that $\lambda_1(\barMst) \ge \deltars$ but~(i.1) does not imply~(i.2): Consider $\EE\{\dir\} = (\log n)^2$ and $\EE\{\dis\} = 9 \log n$, $i \in \mtcvr$. Then $\deltars > 8\log n$ but $\lambda_1(\barMst) = 9 \log n < \sqrt{\Deltar\log n}$ for large~$n$.
    The condition~(iii) assumes that the number of regular agents is proportional to the network size, which is necessary for entry-wise concentration studied in \Cref{cor:number_error}. 
    $\hfill\square$
\end{remark}

We now state the first main theorem, which studies the concentration of $\xgn$ and provides a high-probability bound for the distance between $\xgn$ and $\xstn$.

\begin{theorem}[Concentration of expected final opinions]\label{thm:concentration_states}~\\\indent For $\xgn$ and~$\xstn$ given in~\eqref{eq:bfxg} and~\eqref{eq:bfx*}, respectively, the following results hold.\\\indent
	(i) Under \Cref{asmp:main}~(i.1) and~(ii), it holds that
	\begin{align}\label{eq:concen_x}
		\PP\{\|\xgn - \xstn\| \le \eps_{x,n} \|\zs\|\} &\ge 1 - \eta_{x,n},
	\end{align}
	where
	\begin{align*}
		\eps_{x,n} &= 4 \bigg( \frac{\sqrt{(\Deltars \vee \Deltasr) \log n}}{\deltars} + \frac{2 \sqrt{\Deltar \log n} \|\Psis\|}{\deltars^2}  \bigg) , \\
		\eta_{x,n} &= r_0 n^{1 - \frac{\deltars}{8 \log n}} + 2(1+r_0)n^{-\frac15} + 2 n^{-\frac23},
	\end{align*}
        and $r_0= \nr/n$ is the proportion of regular agents. \\\indent
	(ii) Under \Cref{asmp:main}~(i.2) and~(ii), \eqref{eq:concen_x} holds with
	\begin{align*}
		\eps_{x,n} &= 2 \bigg( \frac{\sqrt{(\Deltars \vee \Deltasr) \log n}}{\lambda_1(\barMst) - 4\sqrt{\Deltar\log n}} + \frac{2 \sqrt{\Deltar \log n} \|\Psis\|}{\lambda_1(\barMst)(\lambda_1(\barMst) - 4\sqrt{\Deltar\log n})}  \bigg), \\
		\eta_{x,n} &= 2(1+r_0)n^{-\frac15} + 2n^{-\frac18}.
	\end{align*}
\end{theorem}

\begin{proof}
	See \Cref{appen:thm:concentration_states}.
\end{proof}

\begin{remark}\label{rmk:stateconcentration}
	The first result indicates that the distance between $\xgn$ and $\xstn$ can be bounded by a quantity depending on expected degrees multiplied by the norm of stubborn agent opinions $\zs$, with probability relying on the network size $n$, the proportion of regular agents $r_0$, and the minimum expected stubborn degree $\deltars$. 
    The second result studies the case where $\deltars = 0$, and replaces $\deltars$ with terms related to $\lambda_1(\barMst)$. Note that $\lambda_1(\barMst) \ge \deltars$ represents the minimum expected influence of stubborn agents on regular agents. A lower bound of $\lambda_1(\barMst)$ can be found in~\cite{manaffam2017bounds}. 
    Neither \Cref{asmp:main}~(i.1) nor~(i.2) guarantees connectivity of the random graph, but they ensure that each connected component is influenced by some stubborn agents w.h.p. 
    \Cref{asmp:main}~(i.1) implies $\ns = \Omega(\log n)$, and $\Deltars\vee\Deltasr \ge \log n$ in~(i.2) implies $\nr\vee\ns = \Omega(\log n)$.
    To derive entry-wise concentration (\Cref{cor:number_error}), a larger lower bound $\nr \ge c_{\tx{r}}n$ (\Cref{asmp:main}~(iii)) is needed. 
    In contrast, $\ns$ needs not be proportional to $n$, as long as the link probability between regular and stubborn agents is large enough. 
    Classic concentration bounds for adjacency and Laplacian matrices~\cite{chung2011spectra,le2017concentration} contain the maximum or minimum expected degree. 
    Our results show that the concentration of expected final opinions depends on the relative magnitude of the two expected degrees. 
    The logarithmic term in the bounds may be removed~\cite{le2017concentration}, as suggested in Section~\ref{sec:simul}. We leave the improvement to future work.
$\hfill\square$
\end{remark}

From \Cref{thm:concentration_states}~(i) we can obtain the following proposition. The proposition provides an entry-wise approximation of $\xgn$ using $\xstn$, lower bounding the number of entries of $\xgn$ that are close to those of $\xstn$.

\begin{proposition}[Entry-wise concentration]\label{cor:number_error}
	For $\eps > 0$ denote $\mtcv^{\eps,n} \Let \{i\in \mtcvr: |\xgn_i - \xstn_i| > \eps\}$. Suppose \Cref{asmp:main}~(ii) and~(iii) hold, and $\deltars = \omega((\log n) \vee \sqrt{(\Deltar \log n)^{1/2} (\Deltars \vee \Deltasr)})$. Then for all $\eps > 0 $, $|\mtcvr \setminus \mtcv^{\eps,n}| = \nr(1-o(1))$ w.h.p.
\end{proposition}
\begin{proof}
    See \Cref{appen:cor:number_error}.
\end{proof}
\begin{remark}\label{rmk:number_error}
    The result shows that most entries of $\xgn$ are close to $\xstn$ if regular agents constitute the majority of the network and the minimum expected stubborn degree is large enough.
    As a consequence, the opinion mean  $\bfl_{\nr}^\tp \xgn/\nr$ is close to its expected version $\bfl_{\nr}^\tp \xstn/\nr$.
    $\hfill\square$    
\end{remark}

Relating $\xgn$ to its expected version $\xstn$ can help us quantify $\xgn$ in more detail. To show this, we first investigate properties of $\xstn$. Let
\begin{small}
\begin{align*}
	\barLg &\Let 
	\begin{bmatrix}
		d^{(\tx{r})}_1 & - a_{12} & \dots & -a_{1,\nr} \\
		\vdots & & & \vdots \\
		-a_{\nr,1} & \dots & -a_{\nr,\nr-1} & d^{(\tx{r})}_{\nr}
	\end{bmatrix}
\end{align*}
\end{small}
\hspace{-1mm}be the Laplacian of the subgraph induced by regular agents, and denote its expectation by $\barLst \Let \EE\{\barLg\}$. Recall that $\Deltarr$ is the maximum of expected regular degrees $\EE\{\dir\}$, $1\le i \le \nr$. When regular agents have larger expected stubborn degrees than regular degrees ($\deltars$ much larger than $\Deltarr$), they can have final opinions close to their stubborn neighbors. In contrast, if regular agents have large expected connectivity among themselves compared with their expected stubborn degrees ($\lambda_2(\barLst)$ much larger than $\Deltars\vee\Deltasr$), they can have final opinions close to each other. The theorem below summarizes these results for the expected final opinions $\xstn$.

\begin{theorem}[Profile of $\xstn$]\label{prop:profilexst} 
        The following results hold for~$\xstn$ given in~\eqref{eq:bfx*}.\\\indent
	(i) (When stubborn agents have relatively large influence) \\\indent
        If $\deltars = \omega(1 \vee \sqrt{\Deltarr (\Deltars \vee \Deltasr)})$, then
	\begin{align*}
	\|\xstn - (\diag(\Psis \bfl_{\ns}))^{-1} \Psis \zs \| = o(\|\zs\|).
\end{align*}\indent
	(ii) (When stubborn agents have relatively small influence)~\\\indent If $\lambda_1(\barMst)  =\omega((\Deltars \vee \Deltasr)^{c_M})$ and $\lambda_2(\barLst) = \omega(1 \vee (\Deltars \vee \Deltasr)^{2-c_M})$ for some $c_M \in (0,1)$, then there exists $\gamma_n \in \RR$ such that $\|\xstn - \gamma_n \bfl_{\nr} \| = o(\|\zs\|)$.
\end{theorem}

\begin{proof}
	See \Cref{appen:prop:profilexst}.
\end{proof}

\begin{remark}
    The first result indicates that, if the influence of stubborn agents is large enough compared with the link strength between regular agents, then entries of $\xstn$ are close to opinions of stubborn agents.
    Thus, polarization may occur if regular agents are connected separately to two groups of stubborn agents holding opposite opinions. 
    In contrast,~(ii) shows that $\xstn$ is close to a consensus vector, if the influence of stubborn agents is much smaller than the link strength between regular agents. Note that $\xstn$ is an expectation and $\Xst(t)$ may not converge a.s.
    $\hfill\square$
\end{remark}

\Cref{prop:profilexst}~(i) only considers the case where every regular agent has positive probability connected to stubborn agents (i.e., $\deltars > 0$). Further results for the case where $\deltars = 0$ can be developed. Denote regular agents that have positive expected stubborn degrees by $1,\dots,n_{\tx{r}1}$, and the rest of the regular agents by $n_{\tx{r}1},\dots, n_{\tx{r}1} + n_{\tx{r}2}$, where $n_{\tx{r}1} + n_{\tx{r}2} = \nr$. That is, $\EE\{\dis\} > 0$ for $1\le i \le n_{\tx{r}1}$ and $\EE\{\dis\} = 0$ for $n_{\tx{r}1} + 1, \dots, \nr$. Hence $\barMst$ given in~\eqref{eq:barMst} and the link probability matrix $\Psis$ can be written in block structures as follows
\begin{align}\nonumber
    ~\\\label{eq:barMstblock}
    \barMst = ~
    \begin{NiceMatrix}[columns-width=2mm]
        \bar{M}^{*(11)} & \bar{M}^{*(12)}   &  ~~n_{\tx{r}1} \\
        \bar{M}^{*(21)} & \bar{M}^{*(22)} & ~~n_{\tx{r}2} 
        \CodeAfter
          \SubMatrix[{1-1}{2-2}]
          \OverBrace[shorten,yshift=2mm]{1-1}{1-1}{n_{\tx{r}1}}
          \OverBrace[shorten,yshift=2mm]{1-2}{1-2}{n_{\tx{r}2}}
          \SubMatrix{.}{1-1}{1-2}{\}}[xshift=2mm]
          \SubMatrix{.}{2-1}{2-2}{\}}[xshift=2mm]
    \end{NiceMatrix}~,\qquad
    \Psis =~
    \begin{NiceMatrix}[columns-width=2mm]
        \Psis_+ &   ~~n_{\tx{r}1} \\
        \bfo_{n_{\tx{r}2},\ns} &  ~~n_{\tx{r}2} 
        \CodeAfter
          \SubMatrix[{1-1}{2-1}]
          \OverBrace[shorten,yshift=2mm]{1-1}{1-1}{\ns}
          \SubMatrix{.}{1-1}{1-1}{\}}[xshift=2mm,extra-height=-0.5mm]
          \SubMatrix{.}{2-1}{2-1}{\}}[xshift=2mm]
    \end{NiceMatrix}~.
\end{align}
The block structures depict the topological relationship between the two types of regular agents.
For agents $1,\dots,n_{\tx{r}1}$, let $\deltars^+ \Let \min_{1\le i \le n_{\tx{r}1}} \{\EE\{\dis\}\}$ be their minimum expected stubborn degree, and $\Deltarr^+ \Let \max_{1\le i \le n_{\tx{r}1}}\{\EE\{\dir\}\}$  be their maximum expected regular degree. These agents have final opinions close to their stubborn neighbors, if they have expected stubborn degrees not only larger than their regular degrees ($\deltars^+$ much larger than $\Deltarr^+$), but also larger than the total link strength between them and the rest of the regular agents ($\deltars^+$ much larger than $\|\bar{M}^{*(21)}\|$). The agents $n_{\tx{r}1} + 1, \dots, \nr$ have final opinions as weighted averages of stubborn opinions, with weights depending on network structure. The following theorem presents the above result, extending \Cref{prop:profilexst}~(i).
\begin{thmbis}
    ((When stubborn agents have large influence and $\deltars = 0$)~\\\indent
    If $\lambda_1(\bar{M}^{*(22)}) = \Omega(1)$ and 
    \begin{align}\label{eq:thm47prime}
    \deltars^+ = \omega\bigg( \max \bigg\{ \|\bar{M}^{*(21)}\| \sqrt{\frac{(\Deltars \vee \Deltasr)}{\lambda_1(\bar{M}^{*(22)})}}, \sqrt{\Deltarr^+ (\Deltars \vee \Deltasr)} , 1 \bigg\} \bigg),
    \end{align}
    then there exists $\tilde{M}^{*} \in \RR^{n_{\tx{r}2} \times n_{\tx{r}1}}$ such that
    \begin{align*}
	\bigg \|\xstn - \begin{bmatrix}
        (\diag(\Psis_+ \bfl_{\ns}))^{-1} \Psis_+ \zs \\ \tilde{M}^{*} \Psis_+ \zs
    \end{bmatrix}
     \bigg \| = o(\|\zs\|).
\end{align*}
\end{thmbis}
\begin{proof}
	See \Cref{appen:prop:profilexst}.
\end{proof}

Combining \Cref{thm:concentration_states,prop:profilexst} yields the following result, quantifying the expected final opinions over the RG-S $\xgn$. The theorem studies two cases: When the minimum expected stubborn degree $\deltars$ is large enough compared with maximum expected degrees, regular agents have final opinions close to their stubborn neighbors. In contrast, assume that the influence of stubborn agents ($\lambda_1(\barMst)$) is large enough for concentration to hold. Then regular agents have similar final opinions, when they have large enough connectivity among themselves compared with their expected stubborn degrees ($\lambda_2(\barLst)$ much larger than $\Deltars\vee\Deltasr$).

\begin{theorem}[Profile of $\xgn$]\label{thm_profileofxg}Suppose that \Cref{asmp:main}~(ii) holds.\\\indent
    (i) If $\deltars = \omega((\log n) \vee \sqrt{(\Deltars \vee \Deltasr)[(\Deltar \log n)^{1/2} \vee \Deltarr] })$. Then w.h.p. \[\|\xgn - (\diag(\Psis \bfl_{\ns}))^{-1} \Psis \zs \| = o(\|\zs\|).\]
    \indent (ii) If there exists $c_M \in (0,1)$ such that $\lambda_1(\barMst) = \omega(\max\{(\Deltars \vee \Deltasr)^{c_M}, (\Deltar $ $\log n)^{1/2},$ $ \sqrt{(\Deltar \log n)^{1/2}(\Deltars \vee \Deltasr)}\})$ and $\lambda_2(\barLst) = \omega((\Deltars \vee \Deltasr)^{2-c_M})$, then  there exists $\gamma_n \in \RR$ such that $\|\xgn - \gamma_n \bfl_{\nr} \| = o(\|\zs\|)$ w.h.p.    
\end{theorem}
\begin{remark}
    \Cref{thm_profileofxg}~(i) shows that polarization of the expected final opinions can occur when the influence of stubborn agents is large, as also shown in~\cite{como2016local}. Similar results hold for the case where $\deltars = 0$, which are omitted due to space limit.    
    \Cref{thm_profileofxg}~(ii) implies that consensus of expected final opinions can appear if the influence of stubborn agents is small, as investigated in~\cite{acemouglu2013opinion}.
    When the influence of stubborn agents is neither large nor small, the expected final opinions exhibit much diversity~\cite{flache2017models,friedkin2015problem} and it is hard to provide universal characterization. But \Cref{thm:concentration_states,cor:number_error} enable approximation of $\xgn$ using $\xstn$, for example, establishing correspondence between expected final opinions and communities for SBMs (see Section~\ref{sec:simul}). 
    It is possible to obtain concentration of opinion variances as in~\cite{acemouglu2013opinion}, by solving the stationary covariance matrix and analyzing its concentration~\cite{xing2022identification}. The assumptions in~(ii) essentially ensure the connectivity of the random graph w.h.p.: \Cref{lem:bernstein} and~\eqref{eq:append_weyl} in \Cref{appen:auxiliary_concentration} yield that $\lambda_2(\barLg) \ge \lambda_2(\barLst) - 4 \sqrt{\Deltarr \log n}$ w.h.p. Note that $(\Deltars \vee \Deltasr)^{2-c_M} \ge \Deltars \ge \lambda_1(\barMst) = \omega(\sqrt{\Deltar \log n})$, so $\lambda_2(\barLg) \ge \omega(\sqrt{\Deltar \log n}) - 4 \sqrt{\Deltarr \log n} = \omega(\sqrt{\Deltarr \log n})$, indicating $\lambda_2(\barLg) > 0$ for large $n$.
	$\hfill\square$
\end{remark}

\subsection{Concentration of Time-Averaged Opinions}\label{sec:transient}

In this subsection we study the concentration of time-averaged opinions $\Smtcg(t) = (\sum_{i=0}^{t-1} \Xmtcg(i))/t$ around the expected final opinions $\xstn$. 
The previous subsection studies the bound for $\|\xgn - \xstn\|$. \Cref{prop:stability} indicates that  the time average  $\Smtcg(t)$ should be close to $\xgn$ when $t$ is large enough. By bounding $\|\Smtcg(t) - \xgn\|$, we obtain the following result. The theorem provides upper bounds for the distance between $\Smtcg(t)$ and $\xstn$, which hold with probability increasing to one as the network size $n$ and the time $t$ increase.

\begin{theorem}[Concentration of time-averaged opinions]\label{thm:concentration_states_time} ~\\\indent 
	(i) Under \Cref{asmp:main}~(i.1) and~(ii), for $\eps_{S,n} > 0$, $t > 2\bar{s}_*/\eps_{S,n}$, it holds that
	\begin{align}\label{eq:concen_S}
		\PP\{\|\Smtcg(t) - \xstn\| \le \sqrt{\nr} \eps_{S,n} + \eps_{x,n} \|\zs\|\} &\ge 1 - \eta_{S,n,t} - \eta_{S,n},
	\end{align}
	where $\eps_{x,n}$ is given in \Cref{thm:concentration_states} (i), and
	\begin{align*}
		\eta_{S,n,t} &= 2 \nr \exp\bigg\{ - \frac{(t\eps_{S,n} - 2 \bar{s}_*)^2}{2t(\bar{s}_*)^2} \bigg\},~\bar{s}_* = \frac{12\sqrt{\nr}c_x \alphast}{\deltars},\\
		 \eta_{S,n} &= r_0 n^{1 - \frac{\deltars}{8\log n}} + 2(1+r_0) n^{-\frac15} + 2n^{-\frac23},~r_0 = \nr/n.
	\end{align*}  \indent
	(ii)  Under \Cref{asmp:main}~(i.2) and~(ii),~\eqref{eq:concen_S} holds for $\eps_{S,n} > 0$, $t > 2\bar{s}_*/\eps_{S,n}$ with $\eps_{x,n}$ given in \Cref{thm:concentration_states} (ii) and
	\begin{align*}
		\eta_{S,n,t} &= 2\nr \exp\bigg\{ - \frac{(t\eps_{S,n} - 2 \bar{s}_*)^2}{2t(\bar{s}_*)^2} \bigg\},~\bar{s}_* = \frac{6\sqrt{\nr}c_x \alphast}{\lambda_1(\barMst) - 4\sqrt{\Deltar \log n} },\\
		\eta_{S,n} &= 2(1+r_0) n^{-\frac15} + 2n^{-\frac18},~r_0 = \nr/n.
	\end{align*}
\end{theorem}

\begin{proof}
	See \Cref{appen:thm:concentration_states_time}.
\end{proof}

\begin{remark}\label{rmk:concentration_state_time}
	The theorem provides high-probability bounds for the distance between time-averaged opinions and expected final opinions over the expected graph. The concentration depends on both network size and time. 
    The error $\eps_{S,n}$ controls the concentration of $\Smtcg(t)$ around $\xgn$. Set $\eps_{S,n} = o(1)$ as $n\to\infty$, and then entry-wise concentration follows from $\sqrt{\nr} \eps_{S,n} + \eps_{x,n}\|\zs\| = o(\sqrt{n})$, under \Cref{asmp:main}~(iii).
    Note that $\eps_{x,n}$ and $\eta_{S,n}$ depend on the network size $n$, and do not vanish for fixed $n$, even if $t\to\infty$. This captures the effect of the network size on concentration. 
    When $n$ is large, the concentration probability depends mostly on $\eta_{S,n,t}$ and time $t$.
    $\hfill\square$
\end{remark}

We conclude this section by connecting the main results to community detection for dynamical processes~\cite{schaub2020blind,xing2023community}. In~\cite{xing2023community} we demonstrate how to recover community structure based on time-averaged opinions, for a gossip model over deterministic weighted graphs. \Cref{thm_profileofxg,thm:concentration_states_time} guarantee that such a method can still work for the gossip model over SBMs, and the communities can be recovered  w.h.p. When the influence of stubborn agents is small, it is possible to derive guarantees for community detection based on transient opinions, following the concentration analysis developed in this paper and~\cite{xing2022transient}.

\section{Simulation}
\label{sec:simul}

In this section, we present simulation to illustrate theoretical findings. First, we compare concentration bounds for expected final opinions provided by \Cref{thm:concentration_states} with numerical experiments. Then, we apply main results to a gossip model over an SBM and validate the conclusions through simulation.

We examine how close the error $\eps_{x,n}$, given in \Cref{thm:concentration_states}~(i), is to the actual value $\eps_{n}^* \Let \|(I - \barQg)^{-1} \barRg - (I - \barQst)^{-1}\barRst\|$, where $[\barQg~\barRg]$ and $[\barQst~\barRst]$~are the expected interaction matrices over a random graph $\mtcg$ and over its expected graph, respectively (see Section~\ref{subsec:gossipoverrandom}). 
We consider a gossip model over $\trgs(\nr,\ns,\Psir,\Psis)$, with network size $n= \nr+\ns$, $\ns=c_{\tx{s}} n$, $\psir_{ii} = 0$, and $\psir_{ij} = \psis_{k,l-\nr} \equiv \psi$ for all $i,j,k\in \mtcvr$ and $l\in\mtcvs$ with $i\not=j$, where $c_{\tx{s}} \in (0,1/2)$ is a constant. 
That is, the random graph model is similar to an Erd{\H{o}}s--R{\'e}nyi model, except that no edges exist between stubborn agents. In the following example, we calculate the bound $\eps_{x,n}$ given in \Cref{thm:concentration_states}~(i) explicitly.
\begin{example}\label{exmp_simul}
    Consider $\trgs(\nr,\ns,\Psir,\Psis)$, with network size $n= \nr+\ns$, $\ns=c_{\tx{s}} n$ with $c_{\tx{s}} \in (0,1/2)$, $\psir_{ii} = 0$, and $\psir_{ij} = \psis_{k,l-\nr} \equiv \psi$ with $i,j,k\in \mtcvr$, $l\in\mtcvs$ and $i\not=j$.
    We calculate the explicit expression of the error bound $\eps_{x,n}$ given in \Cref{thm:concentration_states}~(i). First, we calculate the following quantities: 
    \begin{align*}
        \Deltar &= \max\limits_{i \in \mtcvr}\{\EE\{\di\}\} =  \max\limits_{i \in \mtcvr} \Big\{ \textstyle\sum\limits_{j\in \mtcv}\EE \{a_{ij}\} \Big\} = (n-1)\psi,\\
        \Deltars &= \max\limits_{i \in \mtcvr}\{\EE\{\dis\}\} =   \max\limits_{i \in \mtcvr} \Big\{ \textstyle\sum\limits_{j\in \mtcvs}\EE \{a_{ij}\} \Big\} = \ns \psi = c_{\tx{s}} n\psi,\\
        \Deltasr &= \max\limits_{i \in \mtcvs}\{\EE\{\dir\}\} =   \max\limits_{i \in \mtcvs} \Big\{ \textstyle\sum\limits_{j\in \mtcvr}\EE \{a_{ij}\} \Big\} = \nr \psi = (1-c_{\tx{s}} )n \psi,\\
        \deltars &= \min\limits_{i \in \mtcvr}\{\EE\{\dis\}\} = \ns\psi = c_{\tx{s}} n\psi.
    \end{align*}
    Note that $\|\Psis\| \le \|\Psis\|_1 \vee \|\Psis\|_{\infty} = \Deltasr \vee \Deltars$, where $\|\cdot\|_1$ and $\|\cdot\|_{\infty}$ are the maximum absolute column sum and maximum absolute row sum norms, respectively.
    Therefore, from \Cref{thm:concentration_states}~(i),
    \begin{align*}
        \eps_{x,n} &= 4 \bigg( \frac{\sqrt{(\Deltars \vee \Deltasr) \log n}}{\deltars} + \frac{2 \sqrt{\Deltar \log n} \|\Psis\|}{\deltars^2}  \bigg)\\
        &= 4 \bigg( \frac{\sqrt{(1-c_{\tx{s}})n \psi \log n}}{c_{\tx{s}} n\psi} + \frac{2 \sqrt{(n-1)\psi \log n} \|\Psis\|}{(c_{\tx{s}}n\psi)^2}  \bigg)\\
        &\le 4 \bigg( \frac{\sqrt{(1-c_{\tx{s}} )n\psi \log n}}{c_{\tx{s}} n\psi} + \frac{2 \sqrt{n\psi \log n} (1-c_{\tx{s}} )n\psi}{(c_{\tx{s}}n\psi)^2}  \bigg) \\
        & = \frac{4[c_{\tx{s}}\sqrt{1-c_{\tx{s}}} + 2(1-c_{\tx{s}})]}{c_{\tx{s}}^2} \sqrt{\frac{\log n}{n\psi}} = O\bigg(\sqrt{\frac{\log n}{n\psi}}\bigg).
    \end{align*}
    This result shows that the error bound decreases with the link probability $\psi$.
\end{example}
In the numerical experiment, we set $c_{\tx{s}} = 0.1$ and $\psi = (\log n)^2/n$, run the gossip model with $n$ ranging from $10^2$ to $10^4$, and calculate $\eps_{n}^*$. The value has order $O(1/(\log n))$ as shown in \Cref{fig:bound}, whereas \Cref{thm:concentration_states}~(i) and \Cref{exmp_simul} indicate a bound $\eps_{x,n} = O(1/(\log n)^{1/2})$. 
As discussed in \Cref{rmk:stateconcentration}, it is possible to remove the logarithmic terms in $\eps_{x,n}$, resulting in a tighter bound of the same order as observed in the simulation.

\begin{figure}[tbp]
    \begin{minipage}{0.48\textwidth}
        \centering
        \includegraphics[scale=0.32]{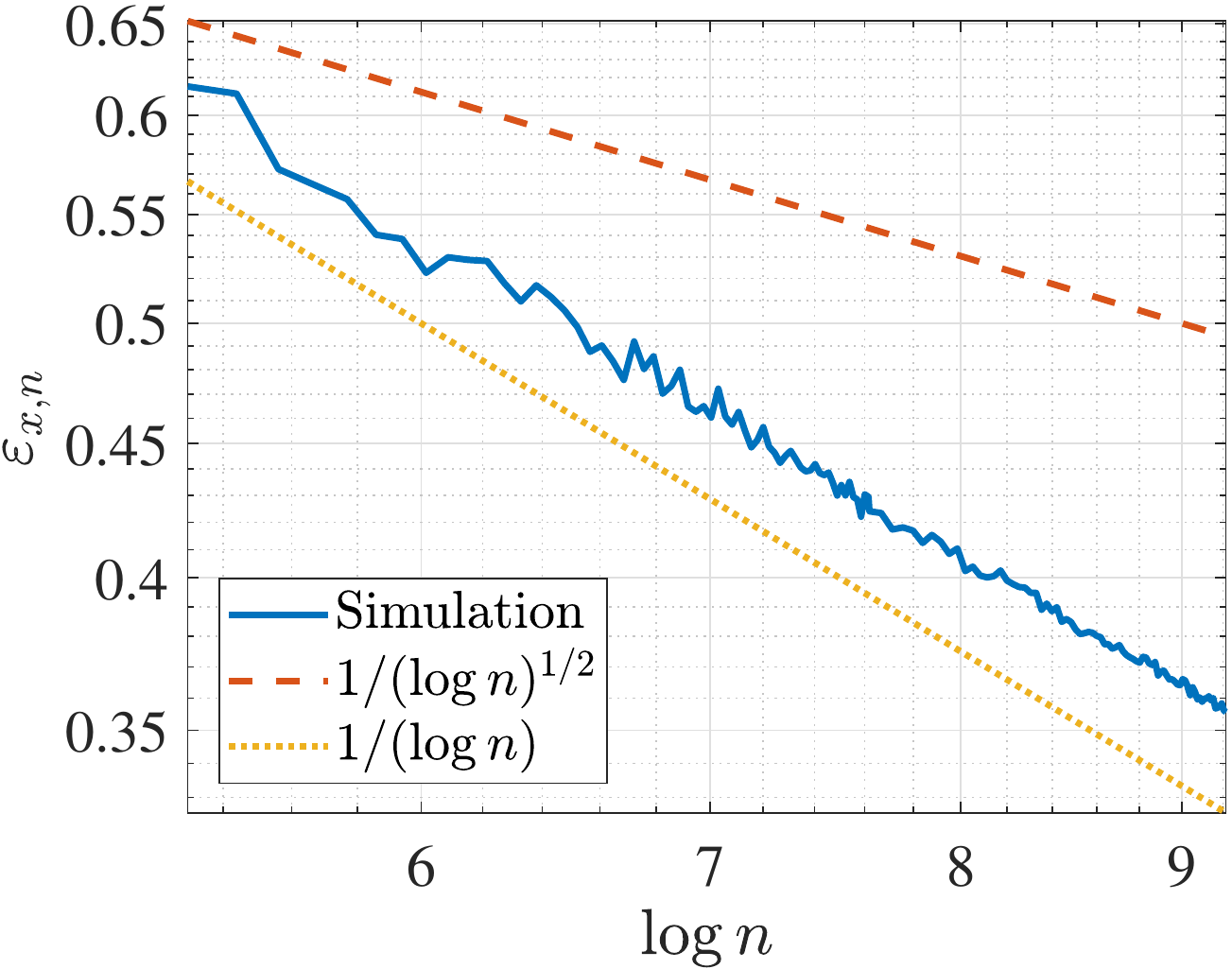}
        \caption{Comparison of the theoretical bound provided by \Cref{thm:concentration_states}~(i) with simulation. A log-log plot is given in the figure.}
        \label{fig:bound}
    \end{minipage}~~
    \begin{minipage}{0.48\textwidth}
        \centering
        \includegraphics[scale=0.3]{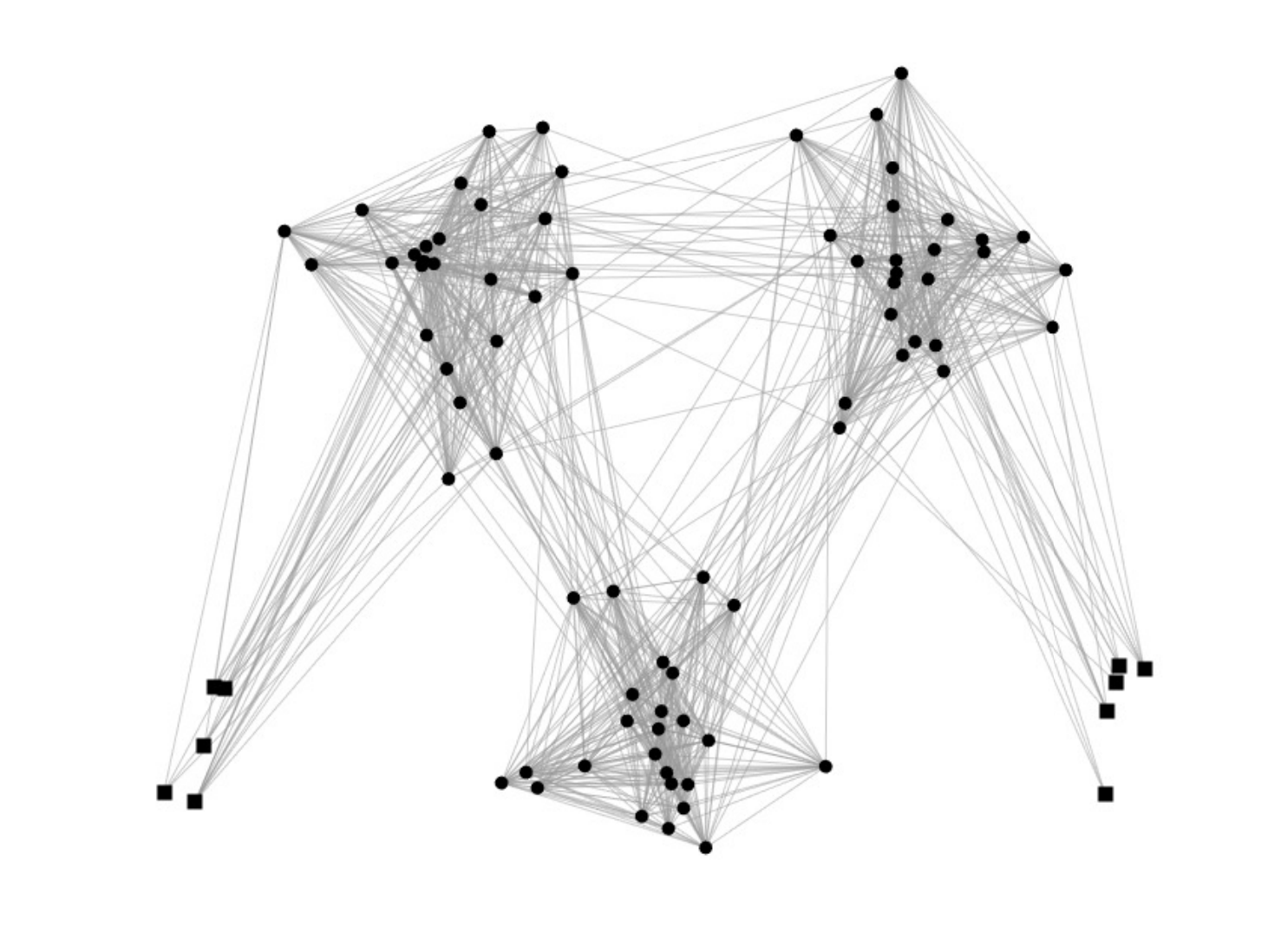}
        \caption{A sample of an SBM. In the graph, dots and squares represent regular and stubborn agents, respectively.}
        \label{fig:sbm}
    \end{minipage}
\end{figure}

\begin{figure*}[t]
	\centering
    \subfigure[\label{fig:phase_transition_large}Large influence ($\gamma = 3.5$).]{
	\begin{minipage}[b]{0.3\textwidth}
		\includegraphics[width=1\textwidth]{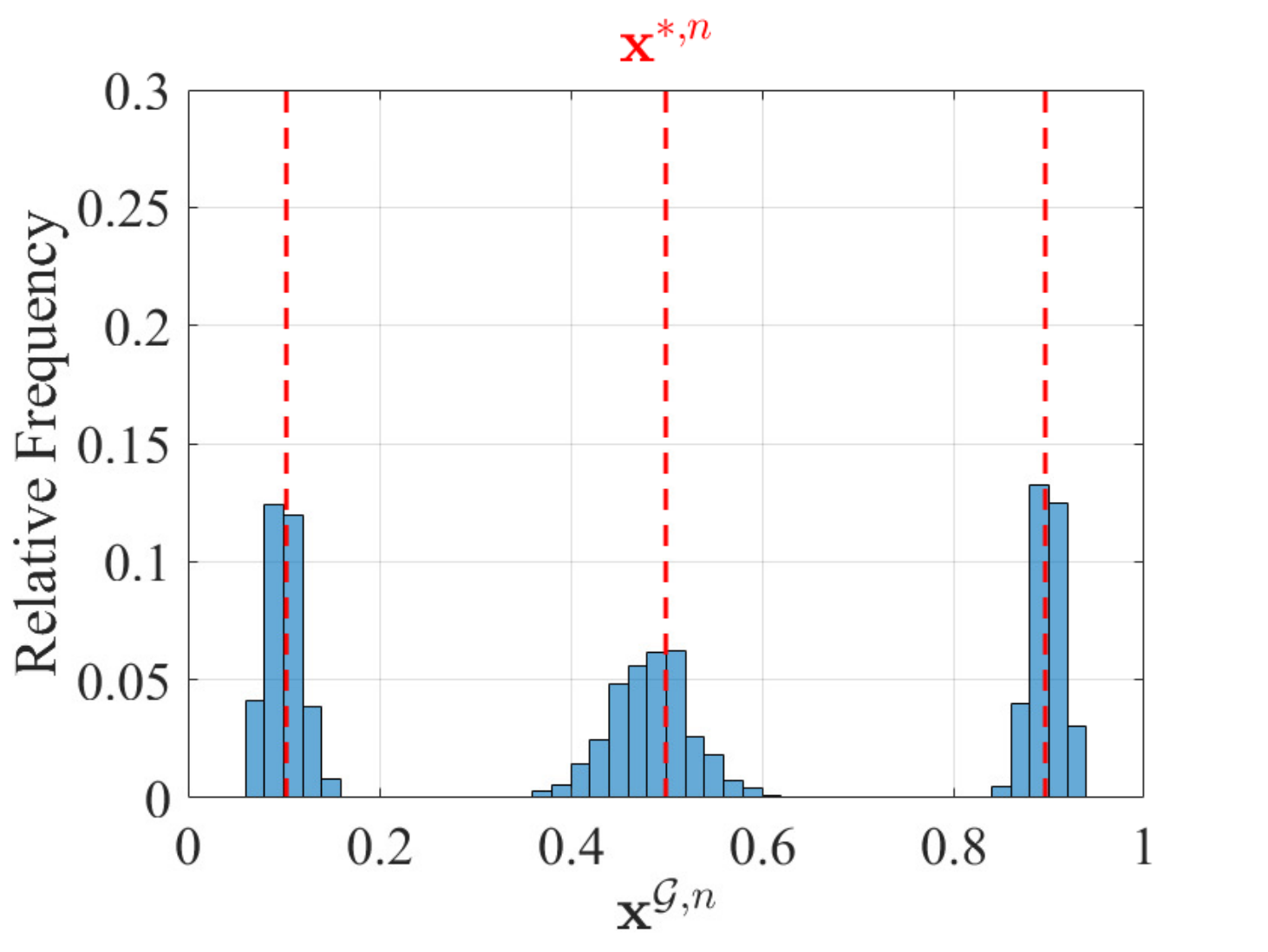} 
	\end{minipage}}
	\subfigure[\label{fig:phase_transition_moderate}Moderate influence ($\gamma = 2$).]{
	\begin{minipage}[b]{0.31\textwidth}
		\includegraphics[width=0.97\textwidth]{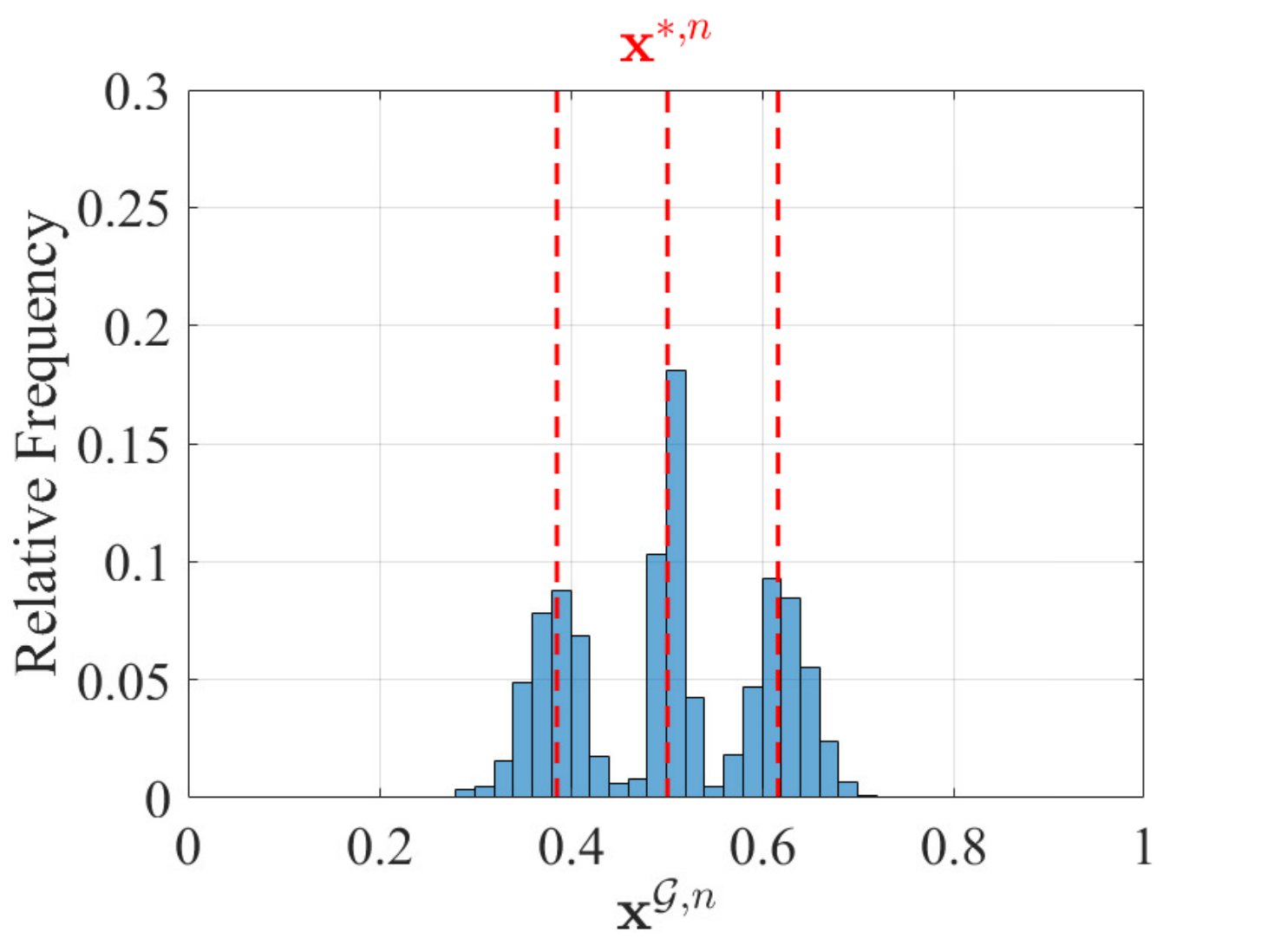} 
	\end{minipage}}
	\subfigure[\label{fig:phase_transition_small}Small influence ($\gamma = 1$).]{
	\begin{minipage}[b]{0.3\textwidth}
		\includegraphics[width=1\textwidth]{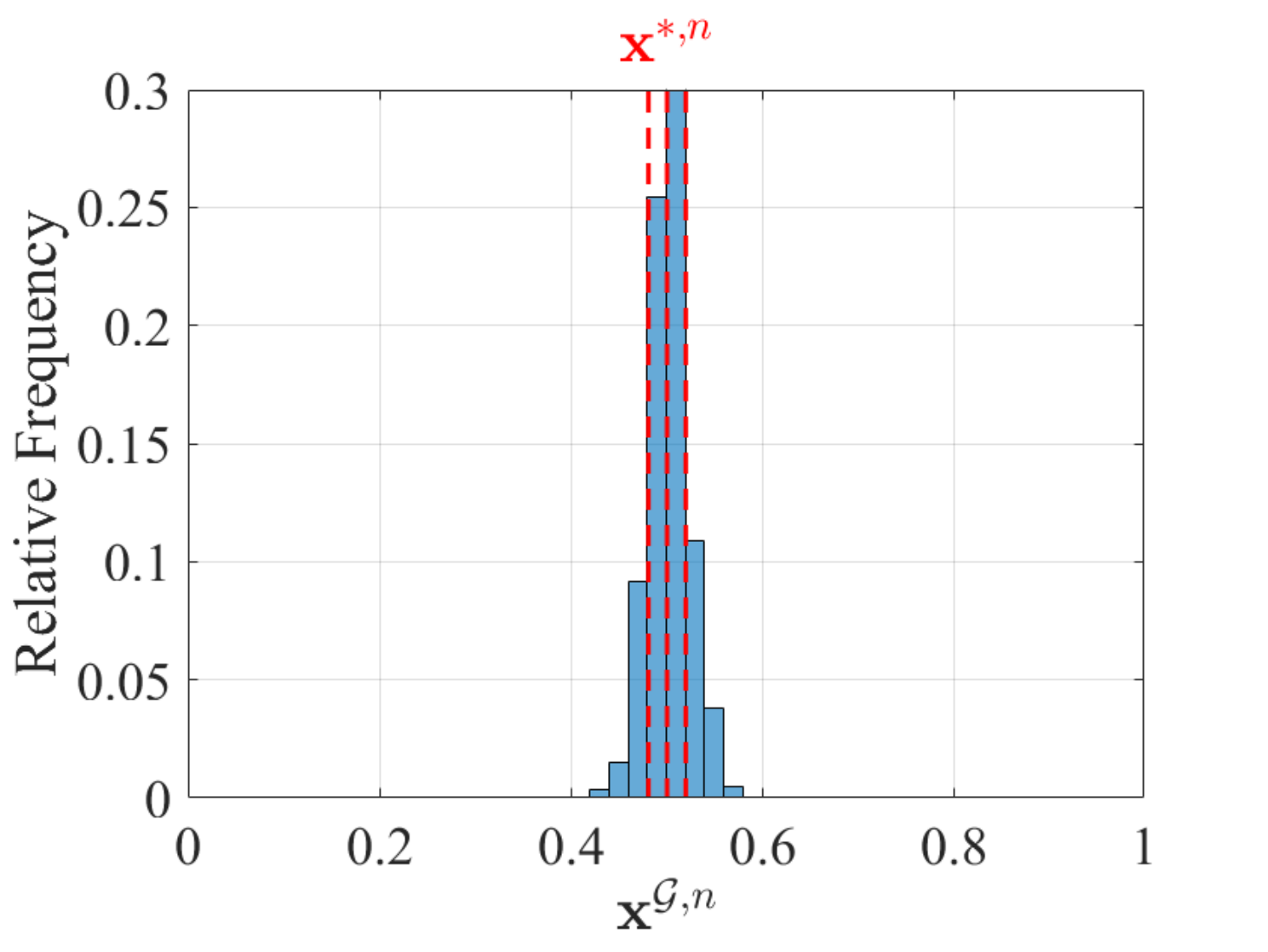} 
	\end{minipage}}
	\caption{\label{fig:phase_transition}The profile of the expected final opinions $\xgn$ under different stubborn influence. The dashed lines represent the three distinct values of $\xstn$ corresponding to the communities.}
\end{figure*}

\begin{figure*}[t]
	\centering 
	\subfigure[\label{fig:vr2_both}$c_{21}^{\tx{(s)}} = c_{22}^{\tx{(s)}} = 1$.]{
		\includegraphics[width=0.3\textwidth]{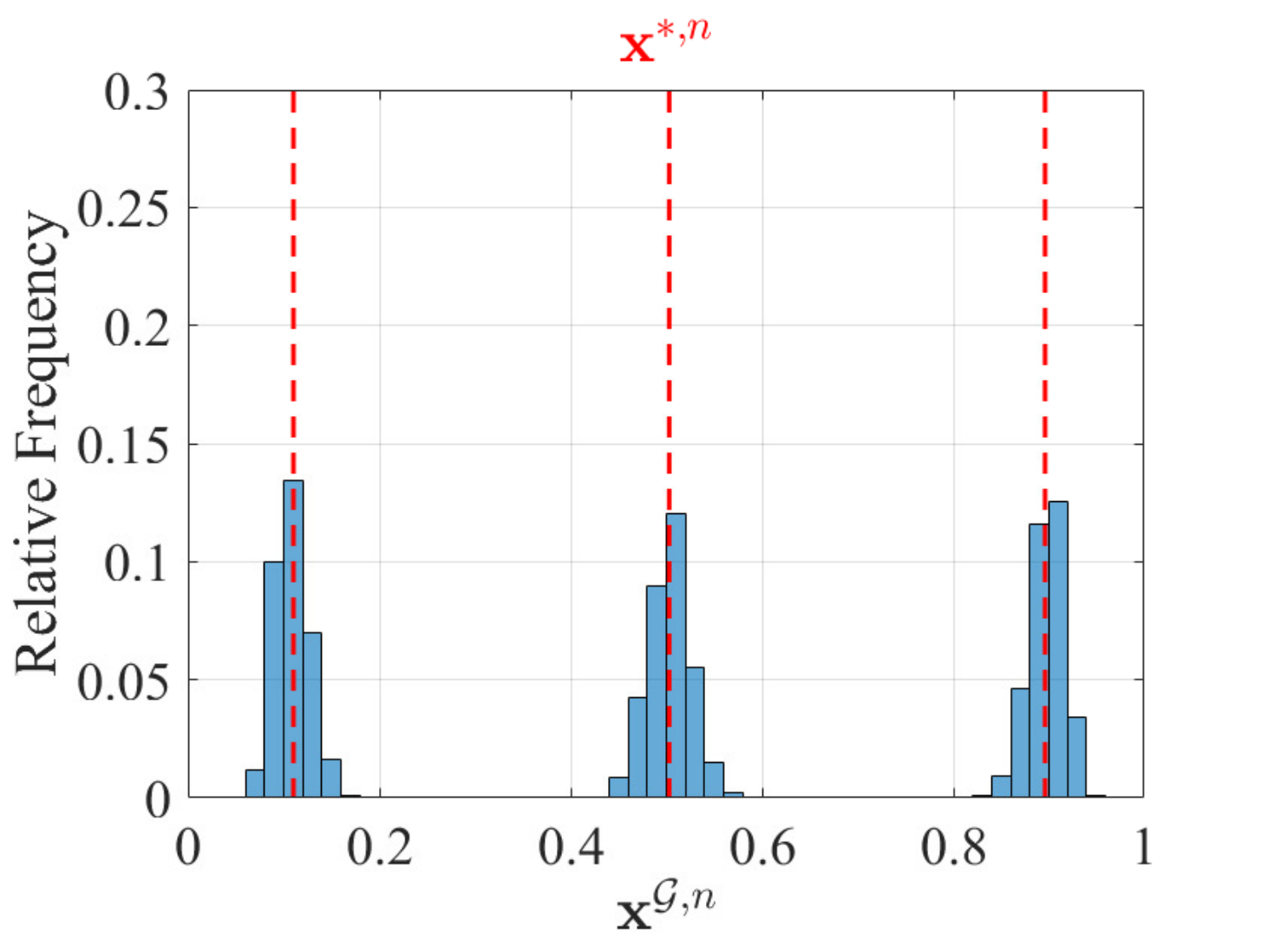} }
	\subfigure[\label{fig:vr2_one}$c_{21}^{\tx{(s)}} =1 $ and $c_{22}^{\tx{(s)}} = 0$.]{
		\includegraphics[width=0.3\textwidth]{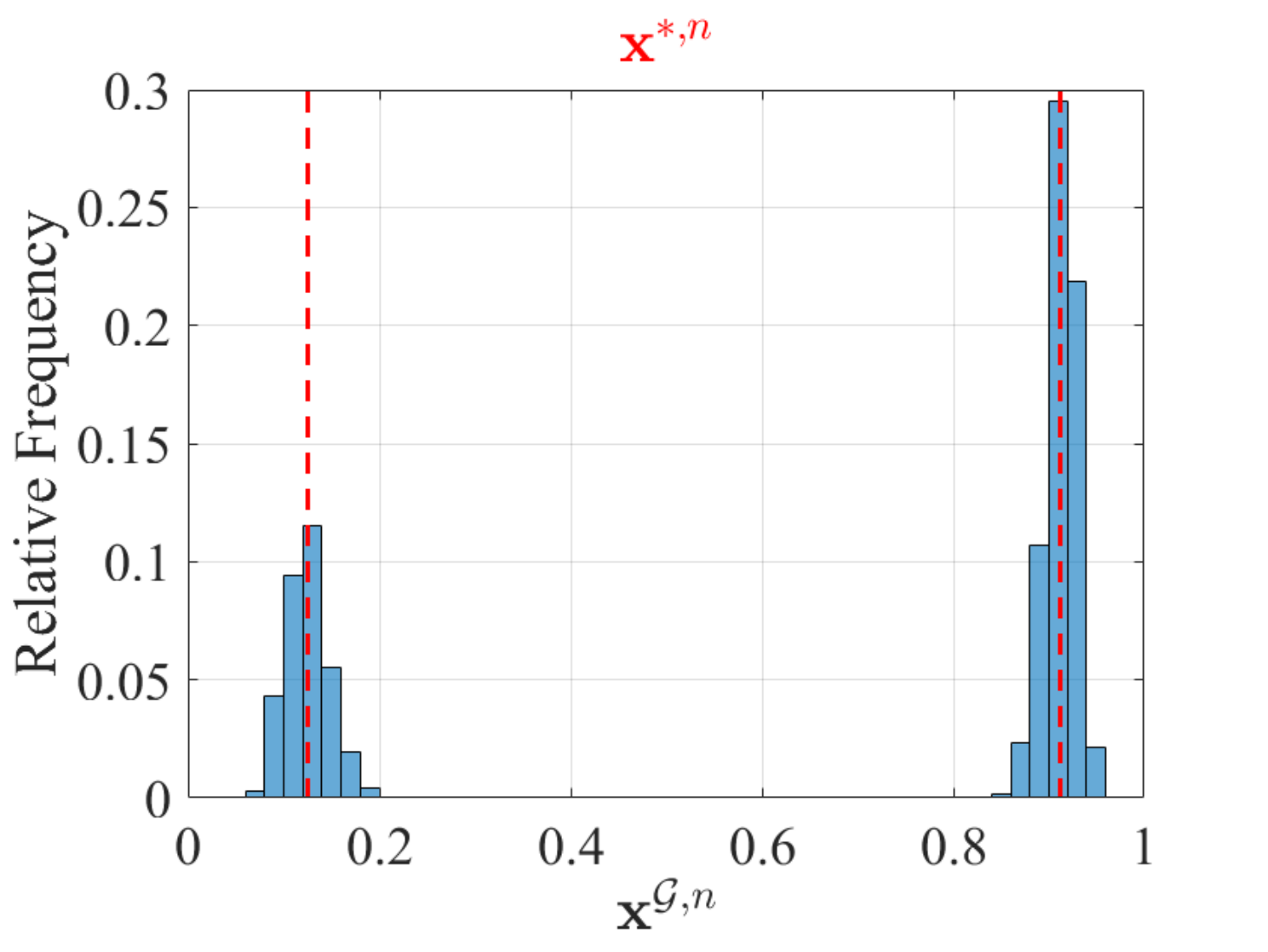} }
	\subfigure[\label{fig:vr2_none}$c_{21}^{\tx{(s)}} = c_{22}^{\tx{(s)}} = 0$.]{
		\includegraphics[width=0.3\textwidth]{xgn_large-eps-converted-to.pdf} }
	\caption{\label{fig:vr2} The profile of the expected final opinions $\xgn$ under different values of $c_{21}^{\tx{(s)}}$ and $c_{22}^{\tx{(s)}}$.}
\end{figure*}

Now we demonstrate the main results by studying the behavior of a gossip model over an SBM. We assume that there are three communities with regular agents $\mtcv_{\tx{r}k} = \{1+(k-1) n_{\tx{r}1}, \dots, k n_{\tx{r}1}\}$,  $1\le k \le 3$, and two communities with stubborn agents $\mtcv_{\tx{s}m} = \{3n_{\tx{r}1} + (m-1) n_{\tx{s}1} + 1, \dots, 3n_{\tx{r}1} + m n_{\tx{s}1}\}$, $m=1,2$. Thus, $|\mtcv_{\tx{r}k}|=n_{\tx{r}1}$, $1\le k\le3$, $|\mtcv_{\tx{s}m}|=n_{\tx{s}1}$, $m=1,2$, and $n = 3n_{\tx{r}1} + 2n_{\tx{s}1}$. 
The network model is $\trgs(3n_{\tx{r}1},2n_{\tx{s}1},\Psir,\Psis)$, where $\psir_{ii} = 0$ for all $i\in \mtcvr$, $\psir_{ij} = (\log n)^{\beta_1}/n \teL p_1$ for $i\not = j \in \mtcv_{\tx{r}k}$ and $1\le k\le 3$, $\psir_{ij} = (\log n)^{\beta_2}/n \teL p_2$ for $i \in \mtcv_{\tx{r}k}$ and $j\in \mtcv_{\tx{r}l}$ with $1\le l\not= k\le 3$, and
\begin{small}
\begin{align*}
    \Psis =  \begin{bmatrix}
        p_3 \bfl_{n_{\tx{r}1}, n_{\tx{s}1}} & \bfo_{n_{\tx{r}1}, n_{\tx{s}1}} \\
        c_{21}^{\tx{(s)}} p_3 \bfl_{n_{\tx{r}1}, n_{\tx{s}1}} & c_{22}^{\tx{(s)}} p_3 \bfl_{n_{\tx{r}1}, n_{\tx{s}1}} \\
        \bfo_{n_{\tx{r}1}, n_{\tx{s}1}} & p_3 \bfl_{n_{\tx{r}1}, n_{\tx{s}1}} 
    \end{bmatrix}
\end{align*}
\end{small}
\hspace{-1mm}with $p_3 =  (\log n)^{\gamma}/n$, $c_{21}^{\tx{(s)}},c_{22}^{\tx{(s)}} \in \{0,1\}$, and $\bfl_{m,n} = \bfl_{m} \bfl_{n}^{\tp}$. That is, the link probability within the same community with regular agents is $p_1$, and the link probability for edges between regular agents in different communities is $p_2$. In addition, regular agents in $\mtcv_{\tx{r}1}$ (resp. $\mtcv_{\tx{r}3}$) have probability $p_3$ linking to stubborn agents in $\mtcv_{\tx{s}1}$ (resp. $\mtcv_{\tx{s}2}$). Agents in $\mtcv_{\tx{r}2}$ have positive probability linking to stubborn agents in $\mtcv_{\tx{s}1}$ (resp. $\mtcv_{\tx{s}2}$) if and only if $c_{21}^{\tx{(s)}} = 1$ (resp. $c_{22}^{\tx{(s)}}=1$). \Cref{fig:sbm} illustrates such an SBM with five communities.
In the experiment, we set $|\mtcv_{\tx{r}k}| = n_{\tx{r}1} = 600$, and $|\mtcv_{\tx{s}m}| = n_{\tx{s}1} = 100$, $1\le k\le 3$, $m=1,2$, so $n=2000$. To fix link probabilities $p_1$ and $p_2$, let $\beta_1  = 2$ and $\beta_2 = 1.1$. For the link probability $p_3$, we consider three cases $\gamma = 3.5,2,$ and $1$, corresponding to large, moderate, and small influence of stubborn agents, respectively. 

We first study the case where $c_{21}^{\tx{(s)}} = c_{22}^{\tx{(s)}} = 0$ (that is, the community $\mtcv_{\tx{r}2}$ does not have any edge connected to $\mtcv_{\tx{s}m}$, $m=1,2$). For each $\gamma$, a network is generated and then fixed. The opinions of stubborn agents in $\mtcv_{\tx{s}1}$ are generated independently and uniformly from $(0.9,1)$ and those in $\mtcv_{\tx{s}2}$ from $(0,0.1)$. The expected final opinions over the SBM and those over the expected graph are calculated according to~\eqref{eq:bfxg} and~\eqref{eq:bfx*}. 
Under the circumstances of interest, it can be proved that there exist $\chi_{k} \in \RR$, $1\le k \le 3$, such that $\xstn_i = \chi_{k}$ for all $i\in \mtcv_{\tx{r}k}$~\cite{xing2023community}. That is, regular agents in the same community have the same expected final opinion. Thus, \Cref{thm:concentration_states} ensures that expected final opinions $\xgn$ in the same community are close, which can be observed in \Cref{fig:phase_transition}.
In addition, large influence of stubborn agents fosters polarization, whereas small influence of stubborn agents results in expected opinions close to each other. These results are in line with theoretical findings given in \Cref{thm_profileofxg}. In the moderate influence case, the expected opinions concentrate around their expected counterparts.
\begin{figure*}[t]
	\centering 
	\subfigure[\label{fig:timeav1}Large influence ($\gamma = 3.5$).]{
	\begin{minipage}[b]{0.3\textwidth}
		\includegraphics[width=1\textwidth]{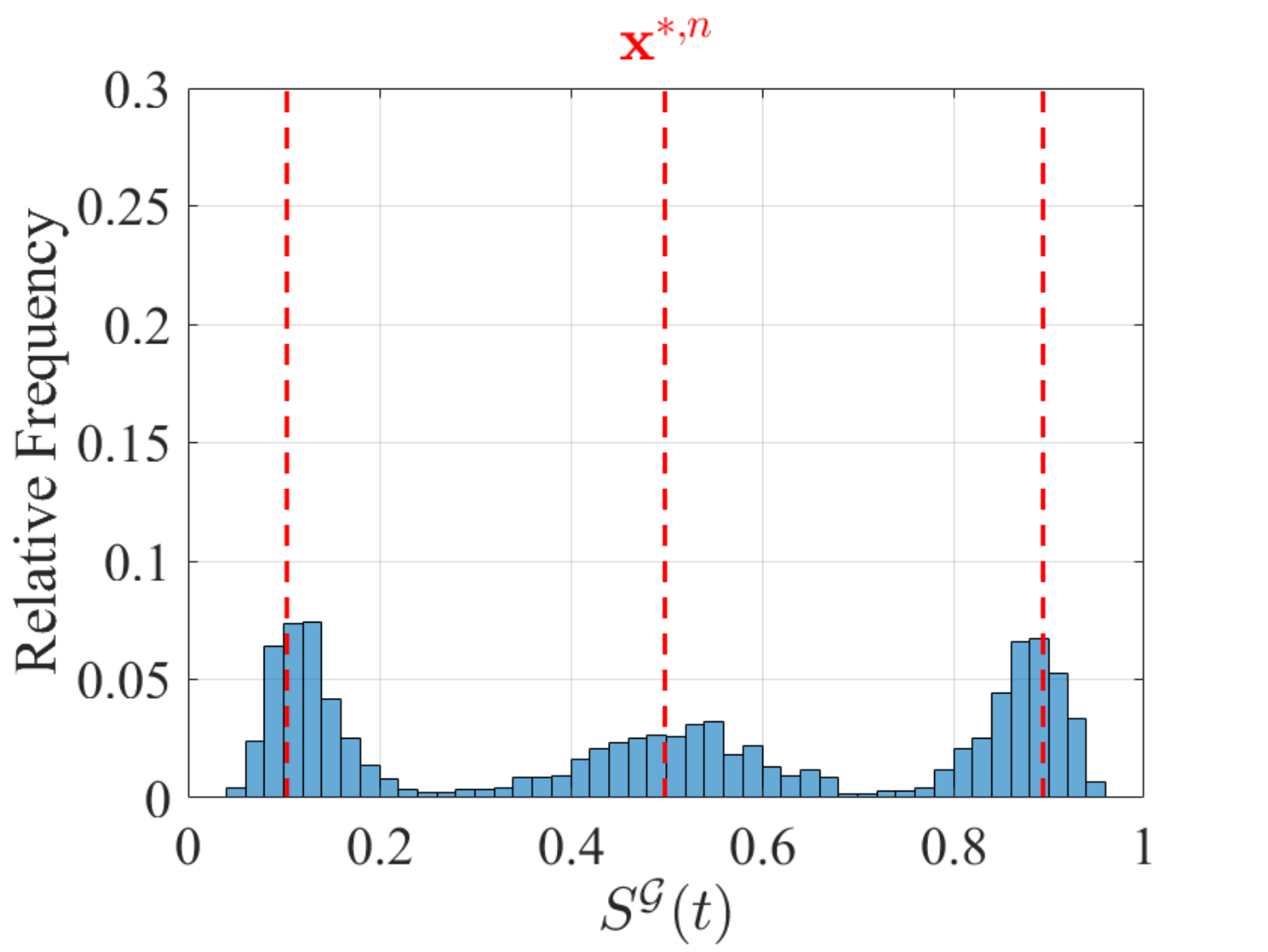} 
	\end{minipage}} 
	\subfigure[\label{fig:timeav2}Moderate influence ($\gamma = 2$).]{
	\begin{minipage}[b]{0.31\textwidth}
		\includegraphics[width=0.97\textwidth]{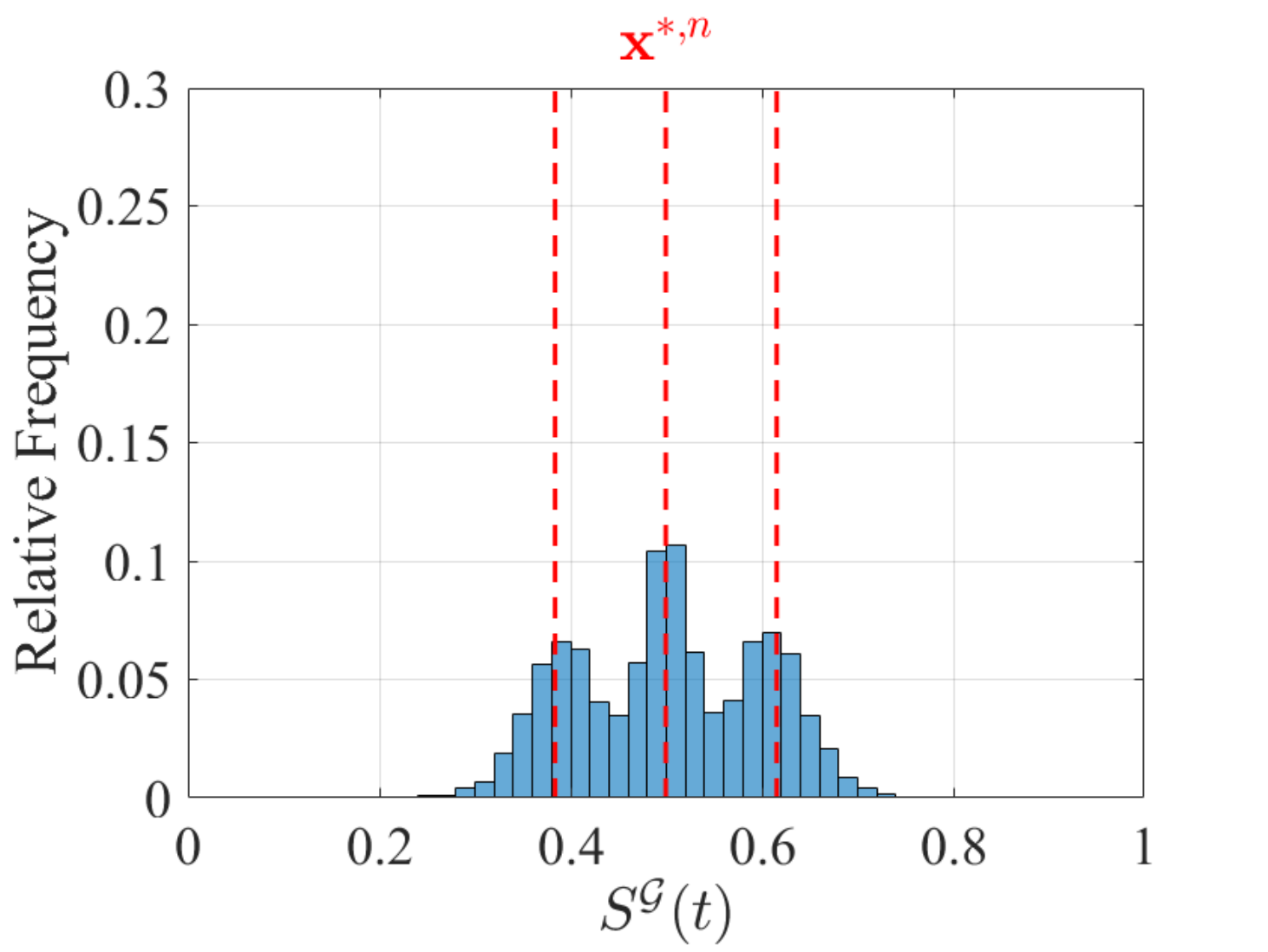} 
	\end{minipage}} 
	\subfigure[\label{fig:timeav3}Small influence ($\gamma = 1$).]{
	\begin{minipage}[b]{0.3\textwidth}
		\includegraphics[width=1\textwidth]{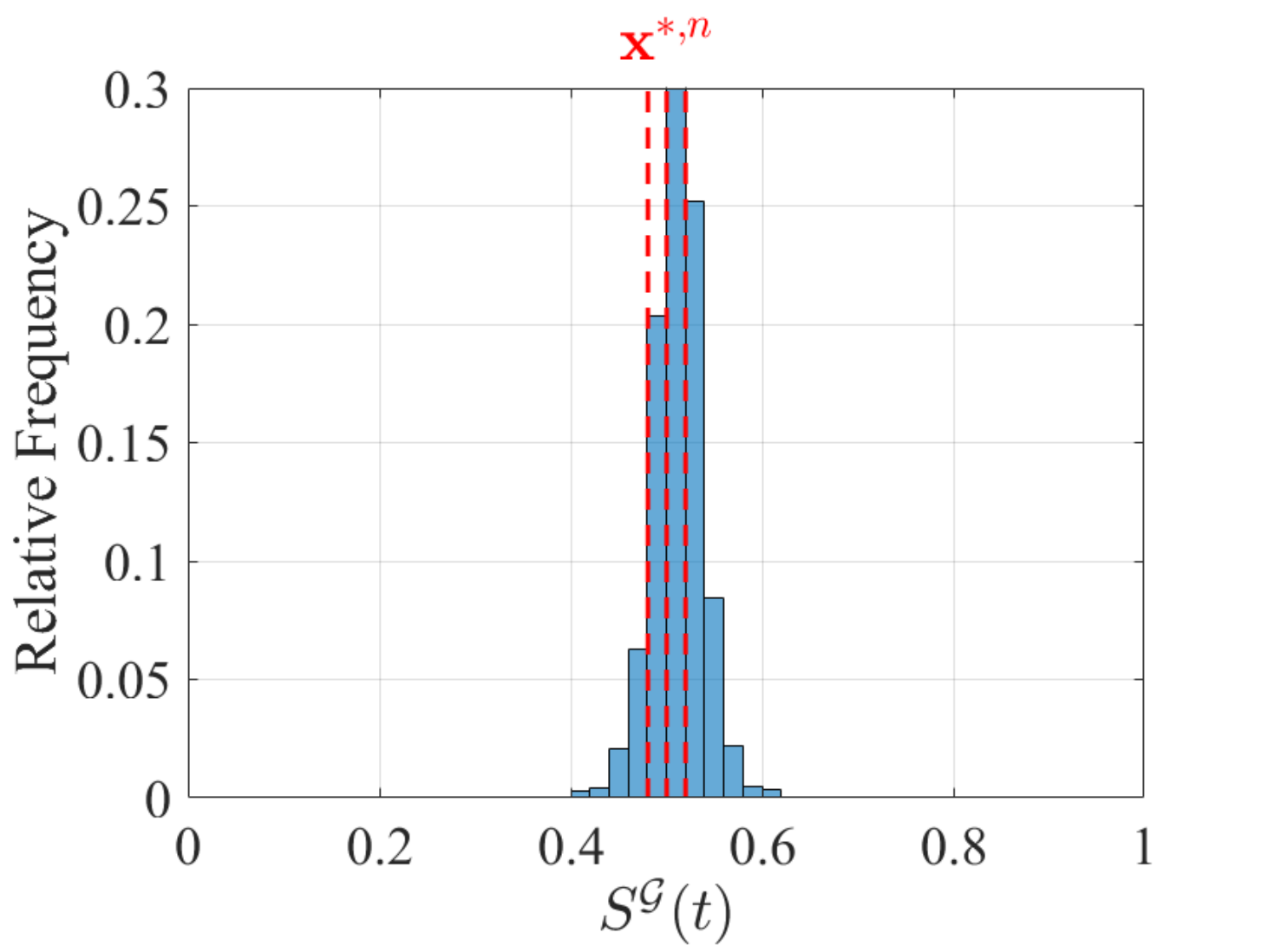} 
	\end{minipage}}
	\caption{\label{fig:timeav} The profile of the time-averaged opinions $\Smtcg(t)$ with $t=5\times 10^4$ under different stubborn influence. The dashed lines represent the three  values of $\xstn$ corresponding to the communities. }
\end{figure*}
Now we examine how edges between $\mtcv_{\tx{r}2}$ and stubborn agents influence $\xgn$. 
Consider three cases: (i) $c_{21}^{\tx{(s)}} = c_{22}^{\tx{(s)}} = 1$, (ii) $c_{21}^{\tx{(s)}} =1 $ and $c_{22}^{\tx{(s)}} = 0$, (iii) $c_{21}^{\tx{(s)}} = c_{22}^{\tx{(s)}} = 0$. 
In case~(i), agents in $\mtcv_{\tx{r}2}$ are connected to $\mtcv_{\tx{s}m}$ with positive probability, $m=1,2$. In case~(ii), they are only connected to $\mtcv_{\tx{s}1}$. In case~(iii), they are not connected to any stubborn agents. We set $\gamma = 3.5$ and generate $\xgn$ in the same way as earlier. \Cref{fig:vr2} shows that the agents in $\mtcv_{\tx{r}2}$ ends in a neutral place in case~(i). However, in case~(ii), agents in $\mtcv_{\tx{r}2}$ have opinions close to $\mtcv_{\tx{s}1}$, as $\mtcv_{\tx{r}2}$ has edges to $\mtcv_{\tx{s}1}$, but not $\mtcv_{\tx{s}2}$. In case~(iii), the expected final opinions of $\mtcv_{\tx{r}2}$ is similar to case~(i), resulting from that agents in $\mtcv_{\tx{r}2}$ have the same number of edges linking to both regular communities.
To illustrate the concentration of the time-averaged opinions $\Smtcg(t)$, we run the gossip model with $c_{21}^{\tx{(s)}} = c_{22}^{\tx{(s)}} = 0$. The initial opinions $X_i(0)$ are generated uniformly from $(0,1)$.  
\Cref{fig:timeav} presents the histogram of $\Smtcg(t)$ under three different values of $\gamma$ with $t = 5 \times 10^4$. We can see that the profile of $\Smtcg(t)$ is similar to $\xgn$ shown in \Cref{fig:phase_transition}, verifying \Cref{thm:concentration_states_time}.

\section{Conclusion}\label{sec:conclusions}
In this paper, we studied concentration of expected final opinions in the gossip model over random graphs, and showed how such concentration can help study the effect of network structure on expected final opinions. We also obtained concentration bounds for time-averaged opinions. 
Future work includes to investigate sharp concentration bounds for the gossip and other models, and to apply the results to community detection problems.

\appendix

\section{Auxiliary Concentration Results}\label{appen:auxiliary_concentration}
In this section, we present auxiliary concentration lemmas from which the main results given in the paper are obtained. These lemmas are consequences of the following standard conclusions in high-dimensional probability theory and matrix analysis.

\begin{lemma}[The Chernoff inequality, Theorems~4.4 and~4.5 of~\cite{mitzenmacher2017probability}] 
    Suppose that $X_1$, $\dots$, $X_n$ are independent Bernoulli random variables such that $\PP\{X_i = 1\} = p_i = 1 - \PP\{X_i=0\}$. Let $X \Let \sum_{i=1}^n X_i$ and $\mu \Let \EE\{X\} = \sum_{i=1}^n p_i$. Then for $0< \delta < 1$,
    \begin{align}\label{eq:append_chernoff1/3}
        \PP\{X \ge (1+\delta)\mu\} \le e^{-\mu\delta^2/3},~
        \PP\{X \le (1-\delta)\mu\} \le e^{-\mu\delta^2/2}.
    \end{align}
\end{lemma}

\begin{lemma}[The matrix Bernstein inequality, Theorem~5.4.1 and Exercise~5.4.15 of~\cite{vershynin2018high}] \label{lem:bernstein}
    Suppose that $Y_1$, $\dots$, $Y_N \in \RR^{n\times n}$ are independent zero-mean random matrices, and are such that $\|Y_i\| \le K$ a.s., $1\le i\le N$. Then for $a \ge 0$, it holds that
    \begin{align}\label{eq:append_bernstein}
        \PP\Big\{\Big\| \sum\nolimits_{i=1}^N Y_i \Big\| \ge a\Big\} \le 2 n \exp \Big\{\frac{-a^2/2}{\sigma^2 + Ka/3}\Big\},
    \end{align}
    where $\sigma^2 = \|\sum_{i=1}^N \EE\{Y_i^2\} \|$. If $Y_1$, $\dots$, $Y_N \in \RR^{m\times n}$ are independent, mean zero, and such that $\|Y_i\| \le K$ a.s. Then for all $a \ge 0$, it holds that
    \begin{align}\label{eq:append_rectan_bernstein}
        \PP\Big\{\Big\| \sum\nolimits_{i=1}^N Y_i \Big\| \ge a\Big\} \le 2 (m+n) \exp \Big\{\frac{-a^2/2}{\sigma^2 + Ka/3}\Big\},
    \end{align}
    where $\sigma^2 = \max\{\| \sum_{i=1}^N \EE\{Y_i^\tp Y_i\} \|, \| \sum_{i=1}^N \EE\{Y_i Y_i^\tp\}\| \}$.
\end{lemma}

\begin{lemma}
    For $A, B\in \RR^{n\times n}$, if $A$ and $B$ are symmetric, then the Weyl inequality holds (Theorem~4.3.1 and~(6.3.4.1) of~\cite{horn2012matrix}):
    \begin{align}\label{eq:append_weyl}
        \max_{1\le i\le n} |\lambda_i(A) - \lambda_i(B)| \le \|A-B\|.
    \end{align}
    If $A$ and $B$ are invertible, then ((5.8.1) of~\cite{horn2012matrix})    \begin{align}\label{eq:append_invertdiffbound}
        \|A^{-1} - B^{-1}\| \le \|A^{-1}\| \|B^{-1}\| \|A-B\|.
    \end{align}
\end{lemma}

First, we derive a concentration bound for the matrix $\barMg$ given in~\eqref{eq:barMg}.

\begin{lemma}[Concentration of $\barMg$]\label{lem:concentration_M} If $\Delta_r \ge \log n$, then $\PP\{\|\barMg - \barMst\| \le \eps_{M,n}\}  \ge  1 - \eta_{M,n}$,
	where $\eps_{M,n} = 4 \sqrt{\Deltar \log n}$, $\eta_{M,n} = 2 r_0 n^{-\frac15}$, and $r_0 = \nr/n$.
\end{lemma}

\begin{proof}
	Decompose $\barMg - \barMst = \sum_{i=1}^{\nr} \sum_{j=i+1}^{n} Y_{ij}$,
	where $	Y_{ij} = (a_{ij} - p_{ij})$ $ (E_{ii} + E_{jj} - E_{ij} - E_{ji})$, $1 \le i < j \le \nr$, and $Y_{ij} = (a_{ij} - p_{ij}) E_{ii}$, $ 1 \le i \le \nr < j \le n$. Here $p_{ij} \Let \EE\{a_{ij}\}$ and $E_{ij} = e_i^{(\nr)} (e_j^{(\nr)})^\tp$, $1\le i,j\le n$.
	Hence $\EE\{Y_{ij}\} = 0$, $\EE\{Y_{ij}^2\} = 2 (p_{ij} - p_{ij}^2) (E_{ii} + E_{jj} - E_{ij} - E_{ji})$ for $1\le i < j\le \nr$, and $\EE\{Y_{ij}^2\} = (p_{ij} - p_{ij}^2) E_{ii}$ for $1 \le i \le \nr <j \le n$. 
    Denote $\bar{Y} \Let \sum_{i=1}^{\nr} \sum_{j=i+1}^n \EE\{Y_{ij}^2\}$, so $\sigma^2 = \|\bar{Y}\| \le 4 \max_{1\le i \le \nr} \{\sum_{j=1}^n p_{ij}\} = 4 \Deltar$. 
    From~\eqref{eq:append_bernstein} and $\|Y_{ij}\|\le 2 = K$, for $a > 0$,
	\begin{align*}
		\PP\{\|\barMg - \barMst\| > a\} \le 2 \nr \exp\Big\{\frac{-a^2}{4(2\Deltar + a/3)}\Big\}.
	\end{align*}
	Set $a = 4 \sqrt{\Deltar \log n}$, and from the assumption $\Deltar \ge \log n$ we have that
	\begin{align*}
		\PP\{\|\barMg - \barMst\| > 4 \sqrt{\Deltar \log n}\} &\le 2 r_0n \exp\Big\{\frac{-4\Deltar \log n}{2\Deltar + 4\sqrt{\Deltar \log n}/3}\Big\}\\
		&\le 2 r_0 n \exp \Big\{ \frac{-4 \log n}{2 + 4/3} \Big\} = 2 r_0 n^{-\frac15}.
	\end{align*}
\end{proof}

As a consequence, we can estimate the deviation of $(\barMg)^{-1}$ from $(\barMst)^{-1}$.

\begin{corollary}[Concentration of $(\barMg)^{-1}$]\label{cor:concentration_Minverse}~\\
\indent (i) If  $\deltars > 8 \log n$, then it holds that
	\begin{align}\label{eq:append_Minv_1}
		\PP\{\|(\barMg)^{-1} - (\barMst)^{-1}\| \le \eps_{M,n}^\prime\} \ge 1 - \eta_{M,n}^\prime,
	\end{align}
	where $\eps_{M,n}^\prime = 2\eps_{M,n}/\deltars^2$, $\eta_{M,n}^\prime =  r_0 n^{1 - \deltars/(8\log n)} + \eta_{M,n}$, $\eps_{M,n}$ and $\eta_{M,n}$ are given in \Cref{lem:concentration_M}, and $r_0 = \nr/n$.

    (ii) If $\lambda_1(\barMst) > \eps_{M,n}$ and $\Delta_r \ge \log n$, then~\eqref{eq:append_Minv_1} holds with $\eps_{M,n}^\prime = \eps_{M,n}/$ $[\lambda_1(\barMst)(\lambda_1(\barMst) - \eps_{M,n})]$ and $\eta_{M,n}^\prime = \eta_{M,n}$.
\end{corollary}

\begin{proof}
	Note that $[\barMst]_{ii} = \EE\{\di\}$, and $[\barMst]_{ij} = - \EE\{a_{ij}\}$. So by the Gershgorin circle theorem,  $ \lambda_{\min}(\barMst) \ge \min_{1\le i \le \nr} \{ \EE\{\di\} -  \EE\{\dir \} \}$ $ = \min_{1\le i \le \nr}$ $ \{ \EE \{ \dis  \} \}$ $ = \deltars $.
	Thus, for symmetric $\barMst$, $(\barMst)^{-1}$ exists when $\deltars> 0$. Hence,
\begin{align}\label{eq:append_Minverse_deltars}
		\|(\barMst)^{-1}\| = \frac{1}{\lambda_{\min}(\barMst)} \le \frac{1}{\deltars}.
	\end{align}
	Similarly, from the Gershgorin circle theorem, it follows that $ \lambda_{\min}(\barMg) \ge \min_{1\le i \le \nr}$ $ \{ \di - \dir \} = \min_{1\le i \le \nr} \{ \dis \} $.
	Using~\eqref{eq:append_chernoff1/3} with $\delta = 1/2$, we obtain that
	\begin{align*}
		\PP\Big\{ \min_{1\le i \le \nr} \{ \dis \} > \frac12 \deltars \Big\} 
		&= 1 - \PP\Big\{ \bigcup_{i=1}^{\nr} \Big[ \dis \le \frac12 \EE\{ \dis \} \Big] \Big\}\\
		&\ge 1 - r_0n e^{-\deltars/8}  = 1 - r_0 n^{1 - \deltars/(8\log n)}.
	\end{align*}
	As a result, with probability at least $1 - r_0 n^{1 - \deltars/(8\log n)}$,
	\begin{align}\label{eq:M_inverse_bound}
		\|(\barMg)^{-1}\| = \frac{1}{\lambda_{\min}(\barMg)} \le \frac{2}{\deltars}.
	\end{align}
	Therefore, from~\eqref{eq:append_invertdiffbound}, with probability at least $1 - r_0 n^{1 - \deltars/(8\log n)} - \eta_{M,n}$,
	\begin{align*}
		\|(\barMg)^{-1} - (\barMst)^{-1}\| \le \|(\barMg)^{-1}\| \|(\barMst)^{-1}\| \|\barMg - \barMst\| \le \frac{2\eps_{M,n}}{\deltars^2}.
	\end{align*}

    To show~(ii), note from~\eqref{eq:append_weyl} that with probability at least $1 - \eta_{M,n}$
    \begin{align}\label{eq:append_MsmallestEig}
	|\lambda_{\min}(\barMg) - \lambda_{\min}(\barMst)| \le \eps_{M,n},
    \end{align}
    so $\lambda_{\min}(\barMg) \ge \lambda_{\min}(\barMst) - \eps_{M,n} > 0$ when $\lambda_{\min}(\barMst) > \eps_{M,n}$. Again from~\eqref{eq:append_invertdiffbound}, 
	\begin{align*}
		\|(\barMg)^{-1} - (\barMst)^{-1}\| &\le \|(\barMg)^{-1}\| \|(\barMst)^{-1}\| \|\barMg - \barMst\| \\
            &= \frac{1}{\lambda_{\min}(\barMg)} \frac{1}{\lambda_{\min}(\barMst)} \|\barMg - \barMst\|\\
            &\le \frac{\eps_{M,n}}{\lambda_{\min}(\barMst)(\lambda_{\min}(\barMst) - \eps_{M,n})},
	\end{align*}
    with probability at least $1 - \eta_{M,n}$.
\end{proof}

Similar to $\barMg$, we can obtain concentration of $\barUg$ given in~\eqref{eq:barMg}.

\begin{lemma}[Concentration of $\barUg$]\label{lem:concentration_U} Suppose that $\Deltars \vee \Deltasr \ge \log n$. Then $\PP\{\|\barUg $ $- \barUst \| \le \eps_{U,n} \} \ge 1 - \eta_{U,n}$,
	where $\eps_{U,n} = 2 \sqrt{(\Deltars \vee \Deltasr) \log n}$ and $\eta_{U,n} = 2 n^{-1/5}$.
\end{lemma}

\begin{proof}
	Decompose $\barUg - \barUst = \sum_{i=1}^{\nr} \sum_{j=\nr+1}^n Y_{ij}^\prime$, 
	where $Y_{ij}^\prime = (a_{ij} - p_{ij}) e_i^{(\tx{r})}$ $ (e_j^{(\tx{s})})^\tp$. Here $p_{ij} = \EE\{a_{ij}\}$, $e_i^{(\tx{r})} \Let e_i^{(\nr)}$, and $e_j^{(\tx{s})} \Let e_{j-\nr}^{(\ns)}$. Hence, $	\|Y_{ij}^\prime\| = |a_{ij} - p_{ij}| \|e_i^{(\tx{r})} (e_j^{(\tx{s})})^\tp\| $ $ \le \|e_i^{(\tx{r})} (e_j^{(\tx{s})})^\tp\| = 1$, and
	\begin{align*}
		\bar{Y}^\prime &\Let \sum\nolimits_{i=1}^{\nr} \sum\nolimits_{j=\nr+1}^n \EE\{(Y_{ij}^\prime)^\tp Y_{ij}^\prime\}\\	
		&=\sum\nolimits_{i=1}^{\nr} \sum\nolimits_{j=\nr+1}^n \EE\{(a_{ij} - p_{ij})^2\} e_j^{(\tx{s})} (e_i^{(\tx{r})})^\tp e_i^{(\tx{r})} (e_j^{(\tx{s})})^\tp \\
		&= \diag\big(\sum\nolimits_{i=1}^{\nr} (p_{i,\nr+1} - p_{i,\nr+1}^2), \dots, \sum\nolimits_{i=1}^{\nr} (p_{i,n} - p_{i,n}^2)\big),\\	
	\|\bar{Y}^\prime\| &\le \max_{\nr+1 \le j \le n} \big\{ \sum\nolimits_{i=1}^{\nr} (p_{ij} - p_{ij}^2) \big\} \le \max_{\nr+1 \le j \le n} \big\{ \sum\nolimits_{i=1}^{\nr} p_{ij}\big\} = \Deltasr.
     \end{align*}
	Similarly, let $\bar{Y}^{\prime\prime} \Let \sum_{i=1}^{\nr} \sum_{j=\nr+1}^n \EE\{Y_{ij}^\prime (Y_{ij}^\prime)^\tp\}$, and then we have that
	\begin{align*}
		\|\bar{Y}^{\prime\prime} \| &\le \big\| \diag \big(\sum\nolimits_{j=\nr+1}^{n} (p_{1j} - p_{1j}^2), \dots, \sum\nolimits_{j=\nr+1}^{n} (p_{\nr,j} - p_{\nr,j}^2)\big) \big\| \\
		&\le \max_{1\le i \le \nr} \big\{ \sum\nolimits_{j=\nr+1}^n (p_{ij} - p_{ij}^2) \big\} \le \Deltars.
	\end{align*}
	Let $\sigma^2 = \Deltars \vee \Deltasr$ and $K=1$, and set $a = 2 \sqrt{(\Deltars \vee \Deltasr) \log n}$. From~\eqref{eq:append_rectan_bernstein}, 
	\begin{align*}
		&\PP\{ \|\barUg - \barUst \| > 2 \sqrt{(\Deltars \vee \Deltasr) \log n} \} \\
		&\le 2 (\ns + \nr) \exp \Big\{ \frac{- 2 (\Deltars \vee \Deltasr) \log n}{( \Deltars \vee \Deltasr) + 2 \sqrt{(\Deltars \vee \Deltasr) \log n}/3} \Big\} \\
		&= 2 n \exp \Big\{ \frac{- 2 \log n}{1 + 2 \sqrt{(\log n)/(\Deltars \vee \Deltasr)}/3} \Big\}  \le 2 n^{-\frac15}.
	\end{align*}
\end{proof}

The preceding concentration bounds are useful in analyzing the distance $\|\xgn - \xstn\| =\| (I-\barQg)^{-1} \barRg \zs - (I-\barQst)^{-1} \barRst \zs \| = \| [(\barMg)^{-1} \barUg - (\barMst)^{-1} \Psis]\zs \|$. But to make sure that $(I-\barQg)^{-1}$ is well-defined, we now study $\barQg$ and $\alphag$.

\begin{lemma}[Bound of $\alphag$ and $\rho(\barQg)$]\label{lem:bound_alpha_rhoQ}~\\\indent
    (i) Suppose that $\deltars > 8 \log n$. Then it holds that 
    \begin{align}\label{eq:append_rhoQ}
        \PP\{[\rho(\barQg) \le \eps_{Q,n} < 1] \cap [\alphag \ge \alphast/2 > 0]\} &\ge 1 - \eta_{Q,n} = 1-o(1),
    \end{align}    
    where $\eps_{Q,n} = 1 - \deltars/(6\alphast)$,  $\eta_{Q,n} = r_0n^{1-\deltars/(8\log n)} + 2n^{-2/3}$, and $r_0 = \nr/n$.

    (ii) Suppose that $\lambda_1(\barMst) > \eps_{M,n}$ and $\Delta_r \ge \log n$. Then~\eqref{eq:append_rhoQ} holds with $\eps_{Q,n} = 1 - (\lambda_1(\barMst) - \eps_{M,n})/(3\alphast)$ and $\eta_{Q,n} = \eta_{M,n} + 2n^{-1/8}$, where $\eps_{M,n}$ and $\eta_{M,n}$ are given in \Cref{lem:concentration_M}.    
\end{lemma}

\begin{proof}
    Applying~\eqref{eq:append_chernoff1/3} with $\delta = 1/2$ yields that
\begin{align}\label{eq:appen:alpha_lowerbound}
	\PP\{\alphag - \alphast \le -\alphast/2 \} \le e^{-\frac{\alphast}{8}}.
\end{align}
    When $\alphast \ge \deltars > 8 \log n > 0$, $e^{-\alphast/8} \le e^{-\log n} = n^{-1}$. If $\alphast \ge \Delta_r \ge \log n > 0$,  $e^{-\alphast/8} \le n^{-1/8}$. Hence, $\alphag \ge \alphast/2 > 0$ w.h.p.

    Note that $I - \barQst = \barMst/(2\alphast)$, so $(I - \barQst)^{-1}$ exists under  conditions of either~(i) or~(ii). Since $\barQg = I - \barMg/(2\alphag)$ is symmetric and positive semi-definite, to show $\rho(\barQg) = \lambda_{\max}(\barQg) < 1$, it suffices to provide a lower bound for $\lambda_1(\barMg/(2\alphag))$.

    First we derive a bound under $\deltars > 8\log n$. From~\eqref{eq:M_inverse_bound}, we know that $\PP \{\lambda_1(\barMg) $ $>  \deltars/2 \} \ge 1 - r_0 n^{1 - \deltars/(8\log n)}$.
In addition, applying~\eqref{eq:append_chernoff1/3} with $\delta = 1/2$ yields that
\begin{align}\label{eq:appen:alpha_upperbound}
	\PP\{ 1/(2\alphag) \le 1/(3\alphast) \} \le \PP\{\alphag - \alphast \ge  \alphast/2\} \le e^{-\frac{\alphast}{12}} \le n^{-\frac23},
\end{align}
so $\rho(\barQg)\le 1 - \deltars/(6\alphast)$ holds with probability at least $1 - r_0 n^{1 - \deltars/(8\log n)} - n^{-2/3}$. 
Thus (i) is proved. Combining~\eqref{eq:append_MsmallestEig},~\eqref{eq:appen:alpha_lowerbound}, and~\eqref{eq:appen:alpha_upperbound} yields~(ii).
\end{proof}

\section{Proof of Theorem~\ref{thm:concentration_states}}\label{appen:thm:concentration_states}
From \Cref{lem:bound_alpha_rhoQ}, $(I-\barQg)^{-1}$ exists w.h.p. and $(I - \barQst)^{-1}$ exists under either \Cref{asmp:main}~(i.1) or (i.2). In either case, $\xgn$ and $\xstn$ are well-defined, and it holds that
\begin{align}\nonumber
	\|\xgn - \xstn\|
	&=\| (I-\barQg)^{-1} \barRg \zs - (I-\barQst)^{-1} \barRst \zs \| \\\nonumber
	&= \Big\| \Big[ \Big( \frac{\barMg}{2\alphag} \Big)^{-1} \frac{\barUg}{2\alphag} - \Big(\frac{\barMst}{2 \alphast} \Big)^{-1} \frac{\Psis}{2 \alphast} \Big] \zs \Big\| \\\nonumber
	&= \| [(\barMg)^{-1} \barUg - (\barMst)^{-1} \Psis]\zs \| \\\nonumber
	&= \| \{ (\barMg)^{-1} (\barUg - \Psis) + [(\barMg)^{-1} - (\barMst)^{-1}] \Psis \} \zs \| \\\nonumber
	&\le (\| (\barMg)^{-1} \| \| \barUg - \Psis \| + \| (\barMg)^{-1} - (\barMst)^{-1} \| \|\Psis \| ) \| \zs \|.
\end{align}
From~\eqref{eq:M_inverse_bound}, \Cref{lem:concentration_U}, \Cref{cor:concentration_Minverse}~(i), and \Cref{lem:bound_alpha_rhoQ}~(i), it holds with probability at least $1 - r_0 n^{1 - \deltars/(8 \log n)} - 2(1+r_0)n^{-1/5} - 2 n^{-2/3}$ that
\begin{align}\label{eq:append_xgnbound1}
	&\|\xgn - \xstn\| 
	\le \Big( \frac{2}{\deltars} \eps_{U,n} + \eps^\prime_{M,n} \|\Psis\| \Big) \|z^{(s)}\| \\\nonumber
	&\qquad\qquad\qquad \le 4 \Big( \frac{\sqrt{(\Deltars \vee \Deltars) \log n}}{\deltars} + \frac{2 \sqrt{\Deltar \log n} \|\Psis\|}{\deltars^2}  \Big) \|z^{(s)}\|.
\end{align}
In this way, we prove (i) of the theorem. The second part follows from~\eqref{eq:append_MsmallestEig}, Lem-ma~\ref{lem:concentration_U}, \Cref{cor:concentration_Minverse}~(ii), and \Cref{lem:bound_alpha_rhoQ}~(ii).

\section{Proof of Proposition~\ref{cor:number_error}}\label{appen:cor:number_error}
From the definition of $\mtcv^{\eps,n}$, $\eps^2|\mtcv^{\eps,n}| $ $\le \sum_i (\xgn_i - \xstn_i)^2 = \|\xgn - \xstn\|^2$. Since $\deltars = \omega(\log n)$, $|\mtcv^{\eps,n}|\le \eps_{x,n}^2\|\zs\|^2/\eps^2 \le \eps_{x,n}^2 c_x^2 n/\eps^2$ w.h.p. Note that $\sqrt{(\Deltar \log n)^{1/2} (\Deltars\vee\Deltasr)} \ge \sqrt{(\deltars \log n)^{1/2} (\Deltars\vee\Deltasr)}$ $ = \omega(\sqrt{(\Deltars\vee\Deltasr) \log n})$, and $\sqrt{(\Deltar \log n)^{1/2} (\Deltars\vee\Deltasr)} \ge \sqrt{(\Deltar \log n)^{1/2} \|\Psis\|}$. \\The conclusion then follows from the expression of $\eps_{x,n}$ in \Cref{thm:concentration_states}~(i), \Cref{asmp:main}~(iii), and $\deltars = \omega(\sqrt{(\Deltar \log n)^{1/2} (\Deltars \vee \Deltasr)})$.

\section{Proof of Theorem~\ref{prop:profilexst}}
\label{appen:prop:profilexst}

\textbf{Proof of \Cref{prop:profilexst}~(i).}
Denote $\bar{S} \Let \diag(d_1^{(\tx{s})}, \dots, d_{\nr}^{(\tx{s})}) = \barMg - \barLg$. Note that $(\EE\{\bar{S}\})^{-1}$ exists when $\deltars>0$, so
\begin{align*}
	\|\xstn - (\EE\{\bar{S}\})^{-1} \Psis \zs \| 
	&\le  \| (\barMst )^{-1} - (\EE\{\bar{S}\})^{-1}\| \| \Psis \| \|  \zs  \|  \\
	&\le  \| (\barMst )^{-1}\|  \|(\EE\{\bar{S}\})^{-1}\| \| \barMst - \EE\{\bar{S}\}\| \| \Psis \| \|  \zs  \| \tag{From~\eqref{eq:append_invertdiffbound}} \\
	&\le \frac{1}{\deltars} \frac{1}{\deltars}  \|\barLst\| \| \Psis \| \|  \zs  \| \tag{From
    ~\eqref{eq:append_Minverse_deltars}} \le \frac{2\Deltarr}{\deltars^2} \| \Psis \| \|  \zs  \|.
\end{align*}
If $\Deltarr = 0$, the conclusion holds trivially. Now suppose $\Deltarr = \Omega(1)$. Then $\Deltars \vee\Deltasr \ge \deltars = \omega(1)$. The conclusion follows from $\deltars = \omega(\sqrt{\Deltarr(\Deltars \vee \Deltasr)})$. $\hfill\square$  

\textbf{Proof of \Cref{prop:profilexst}~(ii).}
If $\Deltars = 0$, $\xstn$ is a consensus vector and the result holds. Now assume $\Deltars = \Omega(1)$. The assumption $\lambda_2(\barLst) = \omega((\Deltars\vee\Deltasr)^{2-c_M}) > 2\Deltars > 0$ ensures that the eigenvalue $\lambda_1(\barMst)$ is simple. By $\xi$ denote the unit eigenvector corresponding to $\lambda_1(\barMst)$. Since $\barMst$ is symmetric, it has orthogonal unit eigenvalues $w^{(2)}$, $\dots$, $w^{(\nr)}$ corresponding to its eigenvalues $\lambda_2(\barMst)\le \dots \le \lambda_{\nr}(\barMst)$. Also $\xi$, $w^{(2)}$, $\dots$, $w^{(\nr)}$ form a basis of $\RR^{\nr}$, and $\xi \xi^\tp + \sum_{j=2}^{\nr} w^{(j)} (w^{(j)})^\tp = I_{\nr}$. So
\begin{align*}
	&\Big\|(\barMst)^{-1} \Psis \zs - \frac{1}{\nr \lambda_1{ (\barMst)}} \bfl_{\nr} \bfl_{\nr}^\tp \Psis \zs \Big\| \\
	&= 
	\Big\|(\barMst)^{-1} \Big (\xi \xi^\tp + \sum\nolimits_{j=2}^{\nr} w^{(j)} (w^{(j)})^\tp\Big) \Psis \zs  - \frac{1}{\nr \lambda_1{ (\barMst)}} \bfl_{\nr} \bfl_{\nr}^\tp \Psis \zs \Big\| \\
	&\le
	\Big\| (\barMst)^{-1} \xi \xi^\tp \Psis \zs  - \frac{1}{\nr \lambda_1 (\barMst)} \bfl_{\nr} \bfl_{\nr}^\tp \Psis \zs \Big\| \\
	&\quad + \big\| (\barMst)^{-1} \big (\sum\nolimits_{j=2}^{\nr} w^{(j)} (w^{(j)})^\tp\big)  \Psis \zs \big\|  \teL (I) + (II).
\end{align*}
Note that $(\barMst)^{-1} \xi = \xi/\lambda_1(\barMst)$, so
\begin{align*}
	(I) &= \Big\|  \frac{1}{\lambda_1 (\barMst)} \xi \xi^\tp \Psis \zs  - \frac{1}{\nr \lambda_1 (\barMst)} \bfl_{\nr} \bfl_{\nr}^\tp \Psis \zs \Big\| \\
	&\le \frac{1}{\lambda_1 (\barMst)} \Big\| \xi\xi^\tp - \frac{1}{\nr} \bfl_{\nr} \bfl_{\nr}^\tp \Big\| \|  \Psis \zs \| 
        \le \frac{2 \|  \Psis\| \| \zs \| \Deltars}{\lambda_1(\barMst) (\lambda_2(\barLst) - 2\Deltars)},
\end{align*}
where the last inequality is obtained from the following lemma with $A = \barMst$, $B = \barLst$, and $\zeta = (\lambda_2(\barLst) - 2\Deltars)/2$, which is a consequence of Theorem~5.5 in Chapter~I and Theorem~3.6 in Chapter~V of~\cite{stewart1990matrix}.
\begin{lemma}
    Let $A, B\in \RR^{n\times n}$ be symmetric, and $\mu$ with corresponding unit eigenvector $u$ (resp. $\nu$ with unit eigenvector $v$) be a simple eigenvalue of $A$ (resp. $B$). Denote $r = A v - \nu v$. If there exists $\zeta>0$ such that the eigenvalues of $A$ except $\mu$ lie outside the interval $[\nu-\zeta,\nu+\zeta]$, then $\|uu^\tp - vv^\tp \|\le \|r\|/\zeta \le \|A-B\|/\zeta$.
\end{lemma}
For $(II)$, it holds that
\begin{align*}
	(II) &= \Big\|\sum\nolimits_{j=2}^{\nr} \frac{1}{\lambda_j(\barMst)} w^{(j)} (w^{(j)})^\tp  \Psis \zs \Big\| \\
	&= 
	\sqrt{\sum\nolimits_{j=2}^{\nr} \frac{[(w^{(j)})^\tp  \Psis \zs ]^2}{\lambda_j^2(\barMst)} } \\
	&\le 
	\frac{1}{\lambda_2(\barMst)} \sqrt{\sum\nolimits_{j=2}^{\nr} [(w^{(j)})^\tp  \Psis \zs ]^2} \\
	&=
	\frac{1}{\lambda_2(\barMst)} \Big\| \sum\nolimits_{j=2}^{\nr} w^{(j)}  (w^{(j)})^\tp  \Psis \zs \Big\|  \\
	&=
	\frac{\| (I - \xi\xi^\tp) \Psis \zs\|}{\lambda_2(\barMst)}   \le 
	\frac{2 \| \Psis \| \| \zs\|}{\lambda_2(\barMst)}.
\end{align*}
Let $\gamma_n = \bfl^{\tp}_{\nr} \Psis \zs/(\nr \lambda_1(\barMst))$. Then under the assumptions of the theorem,
\begin{align*}
    \|\xstn - \gamma_n \bfl_{\nr}\| &\le 
    \Big(\frac{2\Deltars \|\Psis\|}{\lambda_1(\barMst) (\lambda_2(\barLst) - 2 \Deltars)} + \frac{2\|\Psis\|}{\lambda_2(\barMst)} \Big) \|\zs\| \\
    &\le \Big(\frac{2\Deltars (\Deltars\vee\Deltasr)}{\lambda_1(\barMst) (\lambda_2(\barLst) - 2 \Deltars)} + \frac{2(\Deltars\vee\Deltasr)}{\lambda_2(\barLst) - 2\Deltars} \Big) \|\zs\| \\
    &= \Big(\frac{2\Deltars (\Deltars\vee\Deltasr)}{\lambda_1(\barMst) \lambda_2(\barLst)(1-o(1))} + \frac{2(\Deltars\vee\Deltasr)}{\lambda_2(\barLst)(1-o(1))} \Big) \|\zs\| = o(\|\zs\|).
\end{align*}
\indent\textbf{Proof of \Cref{prop:profilexst}~(i)$^\prime$.}
Note that $\deltars^+ > 0$ implies $\lambda_1(\barMst)>0$, so $(\barMst)^{-1}$ exists. Denote
\begin{small}
\begin{align*}
    ~\\
    (\barMst)^{-1} \Let~ \begin{NiceMatrix}[columns-width=2mm]
        \tilde{M}^{(11)} & \tilde{M}^{(12)}   &  ~~n_{\tx{r}1} \\
        \tilde{M}^{(21)} & \tilde{M}^{(22)} & ~~n_{\tx{r}2} 
        \CodeAfter
          \SubMatrix[{1-1}{2-2}]
          \OverBrace[shorten,yshift=2mm]{1-1}{1-1}{n_{\tx{r}1}}
          \OverBrace[shorten,yshift=2mm]{1-2}{1-2}{n_{\tx{r}2}}
          \SubMatrix{.}{1-1}{1-2}{\}}[xshift=2mm]
          \SubMatrix{.}{2-1}{2-2}{\}}[xshift=2mm]
    \end{NiceMatrix}, ~\tilde{M}^* = \tilde{M}^{(21)}.
\end{align*}
\end{small}
\hspace{-1mm}Then we have that
\begin{align*}
    \xstn - \begin{bmatrix}
        (\diag(\Psis_+ \bfl_{\ns}))^{-1} \Psis_+ \zs \\ \tilde{M}^{*} \Psis_+ \zs
    \end{bmatrix}
    = \begin{bmatrix}
        [\tilde{M}^{(11)} - (\diag(\Psis_+ \bfl_{\ns}))^{-1}] \Psis_+ \zs \\ \bfo
    \end{bmatrix},    
\end{align*}
so it suffices to bound $\|[\tilde{M}^{(11)} - (\diag(\Psis_+ \bfl_{\ns}))^{-1}] \Psis_+ \zs\|$. $\lambda_1(\bar{M}^{*(22)}) = \Omega(1)$ implies that $\bar{M}^{*(22)}$ is invertible. From the inverse formula of block matrices \cite{henderson1981deriving}, it follows that $\tilde{M}^{(11)} = [\bar{M}^{*(11)} - \bar{M}^{*(12)} (\bar{M}^{*(22)})^{-1} \bar{M}^{*(21)}]^{-1}$. Hence,
\begin{small}
\begin{align}\nonumber
    &\|\tilde{M}^{(11)} - (\bar{M}^{*(11)})^{-1}\| \\\nonumber
    &= \|[\bar{M}^{*(11)} - \bar{M}^{*(12)} (\bar{M}^{*(22)})^{-1} \bar{M}^{*(21)}]^{-1} - (\bar{M}^{*(11)})^{-1} \| \\\label{eq:appendC:from_invertdiffbound}
    &\le \|[\bar{M}^{*(11)} - \bar{M}^{*(12)} (\bar{M}^{*(22)})^{-1} \bar{M}^{*(21)}]^{-1}\| \| (\bar{M}^{*(11)})^{-1} \| \|\bar{M}^{*(12)} (\bar{M}^{*(22)})^{-1} \bar{M}^{*(21)} \|\\\nonumber
    &= \frac{\|\bar{M}^{*(12)} (\bar{M}^{*(22)})^{-1} \bar{M}^{*(21)} \|}{\lambda_1(\bar{M}^{*(11)}) \lambda_1(\bar{M}^{*(11)} - \bar{M}^{*(12)} (\bar{M}^{*(22)})^{-1} \bar{M}^{*(21)})}  \\\label{eq:appendC:from_weyl}
    &\le \frac{\|\bar{M}^{*(12)} (\bar{M}^{*(22)})^{-1} \bar{M}^{*(21)} \|}{\lambda_1(\bar{M}^{*(11)}) (\lambda_1(\bar{M}^{*(11)}) - \|\bar{M}^{*(12)} (\bar{M}^{*(22)})^{-1} \bar{M}^{*(21)} \|)} \\\label{eq:appendC:from_invertbound2}
    &\le \frac{\|\bar{M}^{*(21)}\|^2/\lambda_1(\bar{M}^{*(22)})}{\lambda_1(\bar{M}^{*(11)}) (\lambda_1(\bar{M}^{*(11)}) - \|\bar{M}^{*(21)}\|^2/\lambda_1(\bar{M}^{*(22)}))},
\end{align}
\end{small}
\hspace{-1mm}where~\eqref{eq:appendC:from_invertdiffbound} follows from~\eqref{eq:append_invertdiffbound}, 
and~\eqref{eq:appendC:from_weyl} from~\eqref{eq:append_weyl}. Similarly we obtain that
\begin{align}\label{eq:appendC:diagonalbound}
    \|(\bar{M}^{*(11)})^{-1} - (\diag(\Psis_+ \bfl_{\ns}))^{-1}\| \le \frac{\|\bar{M}^{*(11)} - \diag(\Psis_+ \bfl_{\ns})\|}{\lambda_1(\bar{M}^{*(11)}) \deltars^+}.
\end{align}
The Gershgorin theorem yields that $\lambda_1(\bar{M}^{*(11)}) \ge \deltars^+$ and $\|\bar{M}^{*(11)} - \diag(\Psis_+ \bfl_{\ns})\| \le 2\Deltarr^+$. The conclusion then follows from~\eqref{eq:appendC:from_invertbound2} and~\eqref{eq:appendC:diagonalbound}. $\hfill\square$  

\section{Proof of Theorem~\ref{thm:concentration_states_time}}\label{appen:thm:concentration_states_time}
Since 
we have derived a bound for $\|\xgn - \xstn\|$ in \Cref{thm:concentration_states},  it suffices to bound $\|\Smtcg(t) - \xgn\|$. To this end, we introduce the following concentration inequality (Lemma~1 of~\cite{xing2023community}) for a Markov chain.
\begin{lemma}[Concentration of time-averaged states]\label{lem:concentration_St_original}
	Consider a discrete-time Markov chain $\{X(t)\}$ taking values on a compact state space $\mtcx$ and having a unique stationary distribution $\pi$. For a function $f: \mtcx \to \RR$ and $\iota \Let \int_\mtcx f(x) \pi(dx)$, denote $g(x) \Let \sum_{t=0}^\infty \EE\{f(X(t)) - \iota|X(0) = x\}$ and $\|g\|_s \Let \sup\{|g(x)|: x\in \mtcx\}$.
	If $\|g\|_s < \infty$, then it holds for $S_f(t) \Let \frac1t \sum_{i=0}^{t-1} f(X(i))$, $\eps > 0$, and $t > 2 \|g\|_s/\eps$ that
	\begin{align*}
		\PP\{|S_f(t) - \iota| \ge \eps\} \le 2 \exp \{ -(t\eps-2\|g\|_s)^2/(2t\|g\|_s^2) \}.
	\end{align*}
\end{lemma}
Conditioned on a graph $\mtcg$, $\rho(\barQg) < 1$ ensures that the gossip model has a well-defined unique stationary distribution with mean $\xgn$, which follows from a standard argument for gossip models (see \cite{acemouglu2013opinion,ravazzi2015ergodic,xing2023community}). \Cref{lem:bound_alpha_rhoQ} ensures that $\rho(\barQg) < 1$ holds with probability at least $1 - \eta_{Q,n}$, where the probability is over the randomness of $\mtcg$.

Now we derive a bound for $\|\Smtcg(t) - \xgn\|$ given a graph $\mtcg$ such that $\rho(\barQg) < 1$. 
Let $f_i(x) = x_i$, $\forall x \in \RR^{\nr}$, $1\le i \le \nr$, and \Cref{lem:concentration_St_original} ensures that for all $\eps > 0$ 
\begin{align}\label{eq:appen:S_preliminary}
	\PP_{\mtcg}\{|S_i^{\mtcg,n}(t) - \bfx_i^{\mtcg,n}| \ge \eps\} \le 2 \exp \{ - (t\eps - 2 \|g_i^{\mtcg,n}\|_s)^2/(2t\|g_i^{\mtcg,n}\|_s^2) \},
\end{align}
where $t >2\|g_i^{\mtcg,n}\|_s/\eps$ and $g_i^{\mtcg,n}$ is the $i$-th component of $G^{\mtcg,n}(x) = \sum_{t=0}^{\infty} \EEg\{X(t) - \xgn | X(0) = x\}$, $x \in \RR^{\nr}$.
Note that $\|g_i^{\mtcg,n}\|_s < \infty$ because for all $x \in \mtcx = [-c_x, c_x]^{\nr}$  
\begin{align*}
	\|G^{\mtcg,n}(x)\| 
	&\le \sum_{t=0}^{\infty} \Big\|(\barQg)^t x - \sum_{i=t}^{\infty} (\barQg)^i \barRg \zs \Big\| 
        = \sum_{t=0}^{\infty} \Big\| (\barQg)^t \Big( x - \sum_{i=0}^{\infty} (\barQg)^i \barRg \zs \Big) \Big\| \\
	&\le \sum_{t=0}^{\infty} \|\barQg\|^t \|x - \xgn\|  = \frac{\|x - \xgn\|}{1 - \rho(\barQg)} \le \frac{2\sqrt{\nr} c_x}{1 - \rho(\barQg)} \teL s_*^{\mtcg,n}.
\end{align*}
Hence $\|g_i^{\mtcg,n}\|_s = \sup_{x \in \mtcx}\{|G_i^{\mtcg,n}(x)|\} \le \sup_{x \in \mtcx}\{\|G^{\mtcg,n}(x)\|\} \le s_*^{\mtcg,n}$. Therefore, from the union bound and~\eqref{eq:appen:S_preliminary}, it follows that for all $\eps > 0$ and $t >2s_*^{\mtcg,n}/\eps$
\begin{align*} 
	\PP_{\mtcg}\{\|\Smtcg(t) - \xgn\| \ge \sqrt{\nr} \eps\} 
        \le  2 \nr \exp \{ - (t\eps - 2 s_*^{\mtcg,n})^2/[2t(s_*^{\mtcg,n})^2] \},
\end{align*}
\Cref{lem:bound_alpha_rhoQ} implies that
\begin{align*}
    \PP \{ \exp \{ - (t\eps - 2 s_*^{\mtcg,n})^2/[2t(s_*^{\mtcg,n})^2] \} \le \exp \{ - (t\eps - 2 \bar{s}_*)^2/[2t(\bar{s}_*)^2] \} \} &\ge 1 - \eta_{Q,n},
\end{align*}
where $\bar{s}_* = 12\sqrt{\nr}c_x \alphast/\deltars$ and $\eta_{Q,n} = r_0n^{1-\deltars/(8\log n)} + 2n^{-2/3}$ if $\deltars > 8\log n$, and $\bar{s}_* = 6\sqrt{\nr}c_x \alphast/(\lambda_1(\barMst)$ $ - \eps_{M,n})$ and $\eta_{Q,n} = \eta_{M,n} + 2n^{-1/8}$ if $\lambda_1(\barMst) > \eps_{M,n}$ and $\Delta_r \ge \log n$.
Denoting $\mtcs_1 = \{ \rho(\barQg) \le \eps_{Q,n}\}$, by the law of total probability,
\begin{align*}
	&\PP\{\|\Smtcg(t) - \xgn\| \ge \sqrt{\nr} \eps\} \\
	&= \PP\{\|\Smtcg(t) - \xgn\| \ge \sqrt{\nr} \eps | \mtcs_1 \} \PP\{\mtcs_1\} + \PP\{\|\Smtcg(t) - \xgn\| \ge \sqrt{\nr} \eps | \mtcs_1^{\tx{c}} \} \PP\{\mtcs_1^{\tx{c}}\} \\
	&\le 2\nr \exp\Big\{ - \frac{(t\eps - 2 \bar{s}_*)^2}{2t(\bar{s}_*)^2} \Big\} \PP\{\mtcs_1\} + \PP\{\mtcs_1^{\tx{c}}\} \le 2\nr \exp\Big\{ - \frac{(t\eps - 2 \bar{s}_*)^2}{2t(\bar{s}_*)^2} \Big\} + \eta_{Q,n}.
\end{align*}
Therefore, the conclusion follows from the above bound and \Cref{thm:concentration_states}.


\bibliographystyle{siamplain}
\bibliography{bibliography}
\end{document}

%% file: ex_shared.tex
\usepackage{lipsum}
\usepackage{amsfonts}
\usepackage{graphicx}
\graphicspath{{figures/}}
\usepackage{subfigure}
\usepackage{epstopdf}
\usepackage{algorithmic}
\ifpdf
  \DeclareGraphicsExtensions{.eps,.pdf,.png,.jpg}
\else
  \DeclareGraphicsExtensions{.eps}
\fi

\newsiamremark{remark}{Remark}
\newsiamremark{hypothesis}{Hypothesis}
\crefname{hypothesis}{Hypothesis}{Hypotheses}
\newsiamthm{claim}{Claim}

\newsiamthm{example}{Example}
\newsiamthm{assumption}{Assumption}

\headers{Concentration in Gossip Opinion Dynamics over Random Graphs}{Y. Xing and K. H. Johansson}

\title{Concentration in Gossip Opinion Dynamics\\ over Random Graphs \thanks{
This work was funded by the Knut and Alice Wallenberg Foundation (Wallenberg Scholar Grant), the Swedish Research Council (Distinguished Professor Grant 2017-01078), and the Swedish Foundation for Strategic Research (CLAS
Grant RIT17-0046).}}

\author{Yu Xing\thanks{Division of Decision and Control Systems, School of Electrical Engineering and Computer Science, KTH Royal Institute of Technology, and Digital Futures, Stockholm, Sweden.
  (\email{yuxing2@kth.se}, \email{kallej@kth.se}).}
\and Karl H. Johansson\footnotemark[2]}

\usepackage{amsopn}
\DeclareMathOperator{\diag}{diag}

\usepackage{enumitem}

\usepackage{nicematrix}

\newenvironment{thmbis}[1]
  {%
   \addtocounter{theorem}{-1}%
   \begin{theorem}}
  {\end{theorem}}
